\documentclass{article}
\usepackage{authblk}
\usepackage[utf8]{inputenc} 
\usepackage[T1]{fontenc}    
\usepackage{url}           
\usepackage{booktabs}       
\usepackage{amsfonts}       
\usepackage{nicefrac}       
\usepackage{microtype}      
\usepackage{amssymb,amsmath,amsthm}
\usepackage{changepage}
\usepackage{xcolor,graphicx}
\usepackage{bm,enumerate}
\usepackage{alltt,listings,tikz}
\usepackage[all]{xy}
\usepackage{hyperref}
\usepgflibrary{shadings}

\usepackage{geometry}
\hypersetup{
    colorlinks,
    linkcolor={black},
    citecolor={black},
    urlcolor={black}
}
\usepackage{epstopdf, epsfig}
\newtheorem{lemma}{Lemma}

\newtheorem{theorem}{Theorem}
\usepackage{tabularx}
\usepackage{bigints}
\usepackage{comment}
\usepackage{soul}
\usepackage{array}
\usepackage{makecell}
\usepackage{natbib}
\usepackage{physics}
\usepackage{todo}
\usepackage{multirow}
\usepackage{tabularx}
\usepackage{enumerate}
\usepackage{mathrsfs}
\usepackage{algorithm}
\usepackage{algorithmic}

\newcommand{\red}[1]{\ifmmode\mathbf{\textcolor{red}{#1}}\else \textbf{\textcolor{red}{#1}}\fi}
\newcommand{\blue}[1]{\ifmmode\mathbf{\textcolor{blue}{#1}}\else \textbf{\textcolor{blue}{#1}}\fi}
\usepackage{xcolor}

\usepackage{mathtools}
\usepackage{cancel}

\newtheorem{remark}[theorem]{Remark}
\newtheorem{claim}[theorem]{Claim}

\newcommand{\ud}{\,\mathrm{d}}
\newcommand{\bfu}{\mathbf{u}}

\newcommand{\bfx}{\mathbf{x}}
\newcommand{\bfE}{\mathbf{E}}

\newcommand{\bfn}{\mathbf{n}}

\newcommand{\bfw}{\mathbf{w}}
\newcommand{\bfI}{\mathbf{I}}
\newcommand{\bfD}{\mathbf{D}}
\newcommand{\bfT}{\mathbf{T}}
\newcommand{\bfJ}{\mathbf{J}}

\newcommand{\bfb}{\mathbf{b}}

\newcommand{\bfnu}{\boldsymbol{\nu}}
\newcommand{\bftau}{\boldsymbol{\tau}}
\newcommand{\vtan}{\bfv_{\mathrm{tng}}}

\newcommand{\divv}{\nabla\!\cdot\!}

\newcommand{\bfTM}{\mathbf{T}_{\mathrm{M}}}

\newcommand{\bff}{\mathbf{f}}

\newcommand{\bfv}{\mathbf{v}}

\newcommand{\grads}{\nabla_S}

\newcommand{\RR}{\mathbb{R}}
\newcommand{\calF}{\mathcal{F}}
\newcommand{\calD}{\mathcal{D}}
\newcommand{\calR}{\mathcal{R}}
\newcommand{\calM}{\mathcal{M}}
\newcommand{\calV}{\mathcal{V}}
\newcommand{\calQ}{\mathcal{Q}}

\newcommand{\GamE}{{\Gamma}}

\newcommand{\cl}{\scriptscriptstyle{\text{CL}}}

\newcommand{\tng}{\scriptscriptstyle{\text{tng}}}

\newcommand{\divS}{\operatorname{div}_S}
\newcommand{\DtS}{D_t^S}
\newcommand{\dSsurf}{\,\ud\mathcal{H}^2}

\title{Energetic variational formulation for electrohydrodynamics of surfactant-laden droplets}

\author[1]{Hangjie Ji}
\author[2]{Jian-Guo Liu}

\affil[1]{Department of Mathematics, North Carolina State University, Raleigh, North Carolina 27695, USA \\
\texttt{hangjie\_ji@ncsu.edu}}

\affil[2]{Department of Mathematics, Duke University, Durham, NC 27708, USA\\
\texttt{jian-guo.liu@duke.edu}}

\date{}

\begin{document}

\maketitle

\begin{abstract}
The coupling between surfactant-laden droplet dynamics and electric fields plays an important role in liquid-handling technologies such as digital microfluidics. We develop an energetic variational framework for the coupled electrohydrodynamics of surfactant-laden droplets, incorporating two-phase Stokes flow, insoluble surfactant transport on a moving interface, and electrostatic effects. Based on Onsager's principle, the governing equations are derived by minimizing the Rayleighian, defined as the sum of the rate of change of the free energy and the dissipation functional, subject to the incompressibility constraint. This formulation simultaneously yields the Stokes equations in each bulk phase, the interfacial stress-balance condition incorporating Marangoni and Maxwell stresses, the electrostatic equation, the surface transport equation for the surfactant concentration, and the moving contact-line dynamics. By replacing the viscous dissipation functional with Rayleigh dissipation, we further derive a reduced model describing the evolution of surfactant-laden droplets driven by motion by mean curvature. For sessile droplets represented as graph surfaces, the model reduces to a one-dimensional coupled electrohydrodynamic model for the liquid height, surfactant concentration, and electrostatic potential. A first-order implicit-explicit numerical scheme is proposed for the reduced system, and numerical results illustrate the coupled effects of surfactant transport and electric fields on driving droplet deformation and migration.
\end{abstract}



\section{Introduction}
\label{sec:intro}
Surfactant-mediated droplet dynamics under electric fields have attracted considerable attention due to their important applications in digital microfluidics, electrowetting systems, and biomedical technologies \citep{kim2001micropumping,abdelgawad2009digital,choi2012digital}. Surfactant molecules adsorbed on fluid interfaces alter local surface tension and generate Marangoni stresses through concentration gradients \citep{Stone1990}, while externally imposed electric fields induce Maxwell stresses that can drive droplet deformation, migration, and interfacial instabilities \citep{vlahovska2019electrohydrodynamics}. The interaction of these mechanisms gives rise to a rich variety of electrically driven phenomena, such as directional transport, wetting/dewetting, and droplet coalescence.

Early studies of electrohydrodynamics focused primarily on the deformation and motion of clean droplets subjected to electric fields, where electrostatic forcing interacts with viscous and capillary effects \citep{Taylor1966,Melcher1969,vlahovska2019electrohydrodynamics}. More recent developments incorporated surfactant transport into electrohydrodynamic systems and demonstrated that surfactant redistribution can significantly modify interfacial stresses through Marangoni effects and strongly influence droplet deformation and stability \citep{poddar2018sedimentation,poddar2019electrical,nganguia2019effects,nganguia2021effects}. Various mathematical frameworks have been proposed to model these coupled phenomena, including formulations based on leaky-dielectric theory \citep{Saville1997,papageorgiou2019film}, phase-field approaches for electrowetting \citep{eck2009phase}, and lubrication-type approximations for thin film droplets \citep{craster2009dynamics,chu2023electrohydrodynamics}. On the analytical side, \cite{fontelos2013uniqueness} established uniqueness and regularity results for a phase-field model for electrowetting. Numerical studies of electrohydrodynamics with surfactant transport have been carried out using level-set formulations \citep{teigen2010influence}, immersed interface methods \citep{nganguia2015immersed}, and boundary integral methods \citep{sorgentone20193d,zhao2021finite}.

The contact-line dynamics of droplets on solid substrates is a fundamental aspect of modeling immiscible two-phase flow systems \citep{snoeijer2013moving}. Diffuse-interface approaches, often implemented through phase-field models \citep{gao2012gradient,gao2014efficient}, have been widely used to regularize the moving-contact-line singularity by replacing the sharp liquid-air interface with a thin transition layer and incorporating generalized boundary conditions \citep{qian2003generalized,xu2018sharp,zhang2022effective}. In contrast, the present work adopts a sharp interface formulation, in which the liquid-air interface and contact line are tracked explicitly. 

Despite substantial progress, several challenges remain in developing a unified framework for surfactant-mediated electrohydrodynamic systems. Existing lubrication-based formulations typically rely on long-wave assumptions and are therefore restricted to thin-film geometries or small interfacial slopes. Furthermore, moving contact-line dynamics within these formulations generally require additional regularization mechanisms, such as precursor layer models or slip boundary conditions \citep{de1985wetting, bonn2009wetting}. Electrostatic effects are also modeled using a variety of approaches, ranging from prescribed surface charge distributions to interfacial charge conservation laws \citep{Saville1997}, causing the coupling mechanisms to depend on the underlying modeling assumptions. Moreover, many existing formulations derive individual governing equations separately, which makes it difficult to systematically enforce consistency with energy dissipation principles. These limitations motivate us to develop a unified framework capable of consistently incorporating surfactant transport, electrostatic effects, and moving contact-line dynamics.

The energetic variational approach, formulated through minimization of the Rayleighian via the \emph{Onsager variational principle} \citep{Doi2011,Doi2013}, provides a natural pathway for constructing such a framework while ensuring thermodynamic consistency and an associated energy-dissipation law. 
For a dissipative continuum system with state variables $\calQ$ and tangent vector $\calV$, the rate of change of the free energy ${\dd\calF}/{\dd t}$ can be expressed as the inner product between the tangent vector $\calV$ and the corresponding unbalanced force $\mathbf{F}$, ${\dd}\calF/{\dd t} = \left<\calV,\mathbf{F}\right>$. The associated quadratic Rayleigh dissipation functional is defined as $\calD = |\calV|^2/{(2\calM)}$, where $\calM$ denotes the mobility coefficient. 
The system evolution is then determined by minimizing the Rayleighian,
\begin{equation}
  \label{eq:R_def_general}
  \calR[\calV;\calQ]
  = \frac{\dd}{\dd t}\calF[\calV;\calQ]
    + \calD[\calV;\calQ]
\end{equation}
with respect to admissible tangent vectors $\mathcal{V}$.
The minimizer of the Rayleighian satisfies the energy-dissipation identity
\begin{equation}
  \label{eq:energy_id}
  \frac{\dd}{\dd t}\calF
  = -2\calD.
\end{equation}

The energetic variational approach has been applied to moving contact-line hydrodynamics \citep{qian2006variational,xu2016variational} and phase-field models for various applications \citep{feng2005energetic,du2009energetic,liu2020variational}. For example, \cite{eck2009phase} developed a phase-field model coupled with the Navier-Stokes system for electrowetting.
Recently, energetic variational formulations have also been applied to contact line and droplet dynamics on structured and textured substrates \citep{GaoLiu2021IFB,GaoLiu2021PhysD,GaoLiu2022SIAM}
. 

In the present work, we apply the Onsager's variational principle to formulate the electrohydrodynamics of surfactant-laden droplets. To achieve this, appropriate forms of the free energy $\calF$, the associated unbalanced forces $\mathbf{F}$, and the tangent vectors $\calV$ must be identified.
One of the main challenges in developing such an energetic variational formulation lies in the consistent coupling of multiple physical mechanisms within a unified framework. In particular, we incorporate surface stresses, Maxwell stresses, moving contact-line dynamics, and surface diffusion along evolving interfaces into the system. The transport of insoluble surfactant, together with system constraints such as the incompressibility condition for the Stokes flow, is incorporated directly into the Rayleighian to ensure the thermodynamic consistency of the resulting coupled system. 
The droplet interface evolves in time and carries insoluble surfactant, leading to additional geometric difficulties. Incorporating Maxwell stresses through variations of the electrostatic energy is also a key aspect of the formulation. Furthermore, the choice of appropriate dissipation functionals is critical for closing the system.

The rest of the paper is structured as follows: Section~\ref{sec:setup} introduces the problem setup. Section~\ref{sec:surfactant_transport} discusses the transport of surfactant along the moving interface via the Reynolds transport theorem.
Section~\ref{sec:energy} defines the free energy functional and computes its
variations. Section~\ref{sec:dissipation_Rayleighian} introduces the dissipation functional and Rayleighian and derives the coupled system through Onsager's principle. Section~\ref{sec:reduced_Rayleigh} introduces a reduced coupled system by replacing the viscous dissipation functional with the Rayleigh dissipation. The resulting coupled system is further reduced to a one-dimensional system by assuming a graph representation of the moving interface, followed by numerical results.
Finally, Section~\ref{sec:conclusion} concludes the paper and provides further discussion.

\section{Problem setup}
\label{sec:setup}
In this section, we present the problem setup for a two-phase Stokes flow in the zero-Reynolds-number regime with insoluble surfactant distributed along a moving interface separating two immiscible phases in the presence of electrostatics. We first introduce the geometric configuration of the problem and then describe the physical quantities involved in the system, including the surface energy density, electrostatic potential, and surface charge density. Similar to previous studies on thin films with electrostatic effects \citep{wray2014electrostatic,wray2022electrostatic}, we neglect intermolecular forces and gravity in the present setting.

\begin{figure}
    \centering
\includegraphics[width=0.6\linewidth]{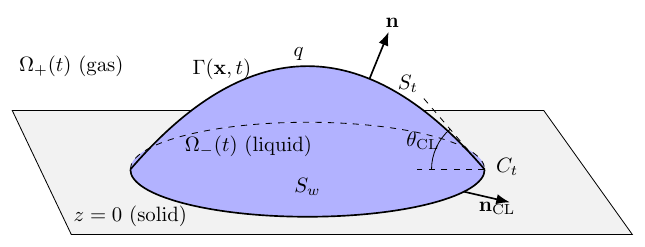}
\includegraphics[width=0.35\linewidth]{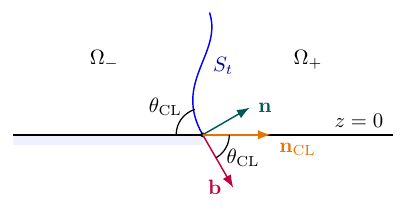}
    \caption{(Left) Schematic of a sessile drop on a solid substrate. The liquid domain $\Omega_-$ rests on the substrate $z = 0$ and is surrounded by the gas domain $\Omega_+$. The free surface $S_t$ meets the substrate at the contact line $\mathcal{C}_t$. The unit normal vector $\bfn$ on $S_t$ points outward from $\Omega_-$ into $\Omega_+$. The dynamic contact angle $\theta_{\cl}$ is measured inside the drop between the substrate and the tangent to $S_t$ at $\mathcal{C}_t$.
    (Right) Detailed contact-line geometry. 
    The co-normal vector $\bfb$, the normal vector $\bfn$, and the contact angle $\theta_{\cl}$ satisfy the relation \eqref{eq:contact_angle_def}.
    }
    \label{fig:geometry}
\end{figure}

Let $\Omega \subset \RR^3_{+}=\{(x,y,z), z>0\}$ be a bounded domain with smooth boundary $\partial\Omega$. For each time $t \ge 0$, let $S_t$ denote a capillary free surface that represents the interface of a sessile droplet placed on the solid substrate $z = 0$. The boundary of $S_t$, denoted by $C_t = \partial S_t$, is assumed to be a smooth closed curve on the substrate and represents the \emph{triple junction} or the \emph{contact line} where the liquid-air interface intersects with the solid substrate. The region enclosed by $C_t$, denoted by $S_w(t)$, represents the wetting domain, namely the liquid-solid interface. The union of the free surface $S_t$ and the wetting domain $S_w(t)$,
 $S_t\cup S_w(t)$, forms a closed surface enclosing the liquid phase
contained in the interior of $\Omega$. This closed surface partitions $\Omega$
into two open subdomains,
$
  \Omega = \Omega_+(t) \cup \Omega_-(t) \cup S_t,
$
where $\Omega_-(t)$ is bounded and represents the liquid domain, while
$\Omega_+(t)$ denotes its complement within
$\Omega$ and represents the gas domains.
This configuration describes a sessile droplet on a flat substrate with contact angle $\theta_{\cl}$, defined as the angle inside the droplet between the substrate $z = 0$ and the tangent plane to the liquid-gas interface $S_t$ at the contact line $C_t$ (see Figure~\ref{fig:geometry} (left)).

Let $\bfu : (\Omega_+\cup \Omega_-) \times [0,T] \to \RR^3$ denote the fluid velocity, $p : (\Omega_+\cup \Omega_-) \times [0,T] \to \RR$ the fluid pressure, and $\phi : \RR^3_{+} \times [0,T] \to \RR$ the electrostatic potential with a prescribed boundary potential $\phi_0$ on the substrate $z = 0$.
Along the interface, we define
$\GamE : S_t \times [0,T] \to [0,\infty)$ as the interfacial surfactant concentration. The interface $S_t$ is the central geometric object of the problem, and $\GamE$ is an intrinsic scalar field defined on $S_t$ representing the surfactant concentration at the interface.
The material properties in both liquid and gas phases are assumed to be constant, including the dynamic viscosity $\mu_\pm > 0$ and electric permittivity $\varepsilon_\pm > 0$ in $\Omega_\pm$.

The contact-line dynamics plays a crucial role in the evolution of the surfactant-mediated free interface. Let $\bfn_{\cl}$ denote the outward unit normal to the contact line $C_t$ within the substrate $z = 0$. The normal speed of the contact line is defined by $v_{\cl} = \bfu|_{C_t}\cdot\bfn_{\cl}$, which represents the outward normal velocity of the contact line along the substrate. 
The unit normal vector on $S_t$ pointing from $\Omega_-(t)$ into $\Omega_+(t)$ is denoted as $\bfn = \bfn(t)$.
Along the contact line $C_t$, let $\bftau_{\cl}$ denote the positively oriented unit tangent vector to $C_t$ on the substrate, and let $\bfb$ denote the outward unit co-normal vector to
  $C_t$, orthogonal to both $\bftau_{\cl}$ and $\bfn$ (see Figure~\ref{fig:geometry} (right)).
  
The contact angle $\theta_{\cl}\in(0,\pi)$ is defined through
\begin{equation}\label{eq:contact_angle_def}
  \mathbf{b}\cdot \bfn_{\cl}=\cos\theta_{\cl}.
\end{equation}
The kinematic conditions governing the contact-line dynamics require that the free surface $S_t$ remain attached to the substrate $z = 0$ and that the contact-line velocity agree in the direction $\bfn_{\cl}$.
Consequently, the contact line moves within the substrate $z = 0$, and its instantaneous velocity can be decomposed as
\begin{equation}
\label{eq:contact_velocity_vector}
  \bfu|_{C_t}=v_{\cl}\,\bfn_{\cl} + w\,{\bftau_{\cl}},
\end{equation}
where $w$ denotes the tangential speed along the contact line.
Combining \eqref{eq:contact_angle_def} and \eqref{eq:contact_velocity_vector}, together with the orthogonality relation $\bftau_{\cl}\perp \bfb$, yields the compatibility condition at the contact line,
\begin{equation}
    \bfu|_{C_t} \cdot \bfb = v_{\cl}\cos\theta_{\cl}.
\label{eq:contact_compatibility}
\end{equation}

The surface energy density associated with the liquid-air interface is denoted by $e(\GamE)$, and the corresponding chemical potential and surface tension are denoted by $\mu(\GamE)$ and $\gamma(\GamE)$, respectively.
Specifically, $e(\GamE)$, $\mu(\GamE)$, and $\gamma(\GamE)$ satisfy the relation \citep{Doi2013,GaoLiu2021PhysD},
\begin{equation}
\gamma(\GamE) = e(\GamE) - e'(\GamE)\GamE, \quad \mu = e'(\Gamma),
\label{eq:surface_energy_chemical}
\end{equation}
We assume that the surface energy density $e(\GamE)$ is a convex function. Consequently, the chemical potential $\mu(\Gamma)$ is an increasing function with respect to $\GamE$, and the surface tension $\gamma(\GamE)$ decreases as the surfactant concentration $\GamE$ increases.
A typical choice of $\gamma(\GamE)$, $e(\GamE)$, and $\mu(\Gamma)$ is obtained from the Langmuir equation, defined as
\begin{align}
    \gamma(\GamE) &= \gamma_0 + \Gamma_s k T \ln\left(1-\frac{\Gamma}{\Gamma_s}\right), \\
    e(\GamE) &= \gamma_0 + kT\left[(\Gamma_s-\GamE)\ln(\Gamma_s-\GamE)+\GamE\ln\GamE-\Gamma_s\ln\Gamma_s\right],
    \label{eq:surface_energy_density}
    \\
    \mu(\GamE) &= k T\ln\frac{\GamE}{\GamE_s-\GamE},
\end{align}
where $\gamma_0$ is the reference surface tension in the absence of surfactant, $\Gamma_s$ is the saturated surfactant concentration, $k$ is the Boltzmann constant, and $T$ is the absolute temperature.
The surface tension coefficients associated with the solid-liquid and solid-gas interfaces, denoted by $\gamma_{SL}$ and $\gamma_{SG}$, are assumed to be constant.

The electrostatic potential $\phi$ determines the electric field through $\bfE = -\grad\phi$
and therefore governs the associated Maxwell stress. 
We assume the absence of free charges in the bulk phases. Consequently, the electric potentials $\phi_{\pm}$ in $\Omega_{\pm}$ satisfy the Laplace equation,
\begin{equation}
    \divv(\varepsilon_{\pm}\grad\phi_{\pm}) = 0
\quad\text{in }\Omega_\pm(t).
\label{eq:laplace}
\end{equation}
The surface charge density $q$ along the interface $S_t$ satisfies the jump condition implied by the Gauss law,
\begin{equation}
    [\varepsilon\partial_n\phi] = \varepsilon_+\partial_n\phi_+ - \varepsilon_-\partial_n\phi_- = q \quad \mbox{on } S_t,
\label{eq:jump_q}
\end{equation}
where the notation $[f]$ denotes the \emph{jump} of the quantity $f$ across the interface $S_t$, namely
\[
  [f] \coloneqq f\big|_{\partial\Omega_+} - f\big|_{\partial\Omega_-},
\]
with the traces taken from the respective sides of the interface. We further assume that both the velocity $\bfu$ and the electrostatic potential $\phi$ are continuous across the interface $S_t$,
\begin{equation}
    [\phi] = 0, \quad [\bfu] = 0.
\label{eq:continuity_phi_u}
\end{equation}

The surface charge density $q(\Gamma)$ is assumed to depend on the surfactant concentration $\Gamma$. A natural choice is the linear relation
\begin{equation}
    q(\Gamma) = A\Gamma,
\label{eq:q_linear_Gamma}
\end{equation}
where $A$ is a constant. This relation expresses the assumption that the surface charge density is proportional to the local surfactant concentration along the interface.

\section{Kinematic description of surfactant transport}
\label{sec:surfactant_transport}
Next, we discuss the transport of insoluble surfactant along the moving interface $S_t$, whose evolution is driven by the velocity field $\bfu$. 
Let $\bfJ$ denote the diffusive flux of surfactant on the interface $S_t$. 
The dynamics of the surfactant concentration $\Gamma$ are influenced by both local changes in the surface area of $S_t$ due to expansion or contraction and the diffusion represented by $\bfJ$.

Let $A_t \subset S_t$ be an arbitrary material surface patch transported by the flow, i.e., each point on $A_t$ moves with velocity $\bfu$.
Conservation of surfactant mass within the patch $A_t$ gives
\begin{equation}
\frac{d}{dt}\int_{A_t} \Gamma \, \ud\mathcal{H}^2 = - \int_{\partial A_t} \bfJ \cdot \bfnu_S \, \ud\mathcal{H}^1 = 
- \int_{A_t} \operatorname{div}_S \bfJ \, \ud\mathcal{H}^2,
\label{eq:surfactant_conservation}
\end{equation}
where $\bfnu_S$ denotes the outward unit co-normal along $\partial A_t$, and the second equality follows from the surface divergence theorem.
Using the surface transport formula \eqref{eq:insoluble_derivative_patch} for the surfactant concentration $\Gamma$ (see Lemma \ref{ref:lemma_surface} in Appendix \ref{sec:appendix_surf}), we obtain
\begin{equation}
    \frac{d}{dt}\int_{A_t} \Gamma \, \ud\mathcal{H}^2
= \int_{A_t} \left(\partial_t \Gamma + \bfu \cdot \nabla_S \Gamma   + \Gamma \, \operatorname{div}_S \bfu\right) \, \ud\mathcal{H}^2.
\label{eq:surfaceTransport}
\end{equation}
Combining \eqref{eq:surfactant_conservation} and \eqref{eq:surfaceTransport} yields
\begin{equation}
\int_{A_t} \left(
\partial_t \Gamma +\bfu \cdot \nabla_S \Gamma + \Gamma \operatorname{div}_S \bfu + \operatorname{div}_S \bfJ
\right)\, \ud\mathcal{H}^2 = 0.
\label{eq:transport_diff_J_integral}
\end{equation}
Since $A_t$ is arbitrary, we deduce the transport equation for the surfactant concentration
\begin{equation}
\partial_t \Gamma + \bfu \cdot \nabla_S \Gamma + \Gamma \operatorname{div}_S \bfu = - \operatorname{div}_S \bfJ.
\label{eq:transport_surfactant}
\end{equation}
This transport equation can also be expressed in a conservative form as 
\begin{equation}
    \partial_t \Gamma + \operatorname{div}_S(\Gamma\bfu) = - \operatorname{div}_S \bfJ.
\label{eq:transport_surfactant_conserved}
\end{equation}

Integrating \eqref{eq:transport_surfactant_conserved} over the entire liquid-air interface, we obtain the rate of change of the total surfactant mass as
\begin{equation}
  \frac{\dd}{\dd t}\int_{S_t}\Gamma \dd \mathcal{H}^2 = -\int_{S_t} \operatorname{div}_S \bfJ ~ \dd \mathcal{H}^2 = -\int_{C_t} \bfJ \cdot \mathbf{b} \,\dd \mathcal{H}^1,
\end{equation}
where we have used the surface divergence theorem.
Since the insoluble surfactant is assumed to remain confined to the liquid-air interface $S_t$, conservation of total surfactant mass requires a zero net flux condition at the contact line $C_t$,
\begin{equation}
\bfJ \cdot \bfb = 0 \qquad \mbox{ on } C_t.  
\label{eq:noflux_bc_Gamma}
\end{equation}

To obtain a geometric interpretation of the transport equation, we decompose the velocity field into normal and tangential components,
\begin{equation}
\bfu = V \bfn + \bfv_{\mathrm{tng}},
\label{eq:velocity_decomp}
\end{equation}
where $\bfv_{\mathrm{tng}}$ denotes the tangential component of the velocity and $V = \bfu \cdot \bfn$ is the normal velocity.
Using the identity
\begin{equation}
\operatorname{div}_S \bfu = \operatorname{div}_S \bfv_{\mathrm{tng}} + \kappa V,
\label{eq:identity_div_u}
\end{equation}
where $\kappa = \operatorname{div}_S\bfn$ denotes the mean curvature of $S_t$,
we rewrite the transport equation \eqref{eq:transport_surfactant} in equivalent form with a geometric interpretation as
\begin{equation}
\partial_t \Gamma + \operatorname{div}_S(\Gamma \bfv_{\mathrm{tng}}) +  \kappa \Gamma V
= -\operatorname{div}_S \bfJ.
\label{eq:surfactant_transport_main}
\end{equation}
Equation \eqref{eq:surfactant_transport_main} reveals three mechanisms that govern the evolution of the surfactant concentration. The tangential advection term $\operatorname{div}_S(\Gamma \bfv_{\mathrm{tng}})$ redistributes the surfactant along the interface, the geometric stretching term $\kappa \Gamma V$ accounts for local expansion or compression of the interface induced by normal motion, and the flux term $-\operatorname{div}_S \bfJ$ represents surfactant diffusion.

The three forms of the transport equation, \eqref{eq:transport_surfactant}, \eqref{eq:transport_surfactant_conserved}, and \eqref{eq:surfactant_transport_main}, 
highlight different perspectives of the surfactant dynamics and will be useful in subsequent derivations. We next turn to the formulation of the free energy and dissipation functionals, which together constitute the Rayleighian framework.

\section{The free energy and its rate of change}
\label{sec:energy}
In this section, we first introduce the free energy of the system, which consists of surface energy and electrostatic energy contributions. We then derive the rate of change of these energy components in subsections ~\ref{sec:variations_surface} and \ref{sec:variations_electric}, respectively. Finally, a decomposition of the unbalanced forces identified by the rate of change of the free energy is presented in subsection~\ref{sec:unbalanced_forces}.

We consider the total free energy of the system,
\begin{equation}
  \label{eq:F_total}
  \calF[S_t, \GamE_t, \phi_t]
  =
  \calF_{\mathrm{surf}}[S_t, \GamE_t]
  + \calF_{\mathrm{elec}}[\phi_t; S_t],
\end{equation}
which consists of the surface energy $\calF_{\mathrm{surf}}$ and the electrostatic energy $\calF_{\mathrm{elec}}$, where $S_t$, $\GamE_t$, and $\phi_t$ represent the interface, surfactant concentration, and electrostatic potential at time $t$, respectively. For notational simplicity, we omit the subscripts $t$ in $\GamE_t$ and $\phi_t$ throughout the rest of the paper unless explicit reference to time dependence is needed.

The surface energy of the droplet accounts for the work required to create and deform the interface and is given by
\begin{equation}
  \label{eq:F_surf}
  \calF_{\mathrm{surf}}[S_t, \GamE]
  = \int_{S_t} e(\GamE)\,\ud\mathcal{H}^2 + (\gamma_{SL}-\gamma_{SG})\int_{S_w}~\dd^2{x},
\end{equation}
where the first term represents the surface energy associated with the liquid-air interface, while the second term represents the surface energy contribution from the solid substrate.

The electrostatic energy stored in the system is given by
\begin{equation}
  \label{eq:F_el}
  \calF_{\mathrm{elec}}[S_t,\phi]
  = \int_{\Omega_+(t)\cup\Omega_-(t)} \frac{\varepsilon(\bfx)}{2}\,|\grad\phi|^2\,\dd^3 x
  +  \int_{S_t} q(\Gamma)\, \phi \, \ud\mathcal{H}^2,
\end{equation}
where $\varepsilon(\bfx) = \varepsilon_\pm$ in $\Omega_\pm$. The first term represents the electrostatic energy stored in the bulk, and the second term accounts for the electrostatic energy associated with the surface charge density on the interface $S_t$.

Combining \eqref{eq:F_surf} and \eqref{eq:F_el}, the total free energy is given by
\begin{equation}
  \label{eq:F_combined}
  \calF[S_t, \GamE, \phi]
  = \int_{S_t} e(\GamE) \,\ud\mathcal{H}^2
  +   (\gamma_{SL}-\gamma_{SG})\int_{S_w}~\dd^2{x} 
  + \int_{\Omega_+(t)\cup\Omega_-(t)} \frac{\varepsilon}{2}|\grad\phi|^2\,\dd^3 x +  \int_{S_t} q(\Gamma)\, \phi \, \ud\mathcal{H}^2.
\end{equation}

\begin{remark}[Stokes regime]
Unlike the electrowetting model considered by \cite{eck2009phase}, the present work focuses on the Stokes regime with zero Reynolds number, where inertial effects are neglected. Consequently, the kinetic energy $\frac{1}{2}\int_\Omega \rho|\bfu|^2\dd x$ does not contribute to the free energy functional $\calF$. Instead, viscous effects are incorporated through the
dissipation functional; see Section~\ref{sec:dissipation}.
\end{remark}

\subsection{Rate of change of the surface energy}
\label{sec:variations_surface}

We now compute the rate of change of the surface free energy $\calF_{\mathrm{surf}}$ defined in \eqref{eq:F_surf} over time.
Differentiating the first term in $\calF_{\mathrm{surf}}$ yields
\begin{align}
\frac{\ud}{\ud t} \int_{S_t} e(\GamE)\, \ud\mathcal{H}^2
&= \int_{S_t} \left( e'(\Gamma)(\partial_t \Gamma +\bfu\cdot \nabla_S \Gamma )
+ e(\Gamma)\operatorname{div}_S \bfu\right) \, \ud\mathcal{H}^2,
\end{align}
where we have used the geometric identity
\begin{equation}
\frac{d}{dt}|A_t| = \int_{A_t} \operatorname{div}_S \bfu \, \ud\mathcal{H}^2
\label{eq:geometric_identity_Jacobi_0}
\end{equation}
for any surface patch $A_t$.
Using the transport equation \eqref{eq:transport_surfactant},
we obtain
\begin{align}
\frac{\ud}{\ud t} \int_{S_t} e(\GamE)\, \ud\mathcal{H}^2
&= \int_{S_t} \big(e(\Gamma)-\Gamma e'(\Gamma) \big) \operatorname{div}_S \bfu \, \ud\mathcal{H}^2- \int_{S_t} \mu(\Gamma) \operatorname{div}_S \bfJ \, \ud\mathcal{H}^2 \nonumber \\
&= \int_{S_t} \gamma(\Gamma) \operatorname{div}_S \bfu \, \ud\mathcal{H}^2
+ \int_{S_t} \bfJ \cdot \nabla_S\mu\, \ud\mathcal{H}^2,
\label{eq:rate_surface_LG}
\end{align}
where we have used the definition of the chemical potential $\mu = e'(\Gamma)$ together with the no-flux boundary condition \eqref{eq:noflux_bc_Gamma}.
Combining \eqref{eq:rate_surface_LG} with
$
\frac{\ud}{\ud t} \int_{S_w}~\dd^2{x} = \int_{C_t}v_{\cl}\, \ud\mathcal{H}^1,
$
the rate of change of the surface free energy \eqref{eq:F_surf} becomes
\begin{equation}
\frac{\ud}{\ud t} \calF_{\mathrm{surf}}=\int_{S_t} \gamma(\Gamma) \operatorname{div}_S \bfu \, \ud\mathcal{H}^2+ (\gamma_{SL}-\gamma_{SG})\int_{C_t}v_{\cl}\, \ud\mathcal{H}^1 + \int_{S_t} \bfJ \cdot \nabla_S\mu \, \ud\mathcal{H}^2.
\label{eq:rate_surf_energy}
\end{equation}

To obtain a geometric interpretation of this rate of change, we use the identity \eqref{eq:identity_div_u} together with the surface divergence theorem to obtain
\begin{equation}
\int_{S_t} \gamma(\Gamma) \operatorname{div}_S \bfu \, \ud\mathcal{H}^2
 = \int_{S_t} (\gamma \kappa V - \bfv_{\tng} \cdot \nabla_S \gamma) \, \ud\mathcal{H}^2 + 
 \int_{C_t} \gamma \mathbf{b}\cdot \bfu\,\ud\mathcal{H}^1.
 \label{eq:div_surface_1}
\end{equation}
From $\mathbf{b}\perp \bftau_{\cl}$ and the definition of the contact angle in \eqref{eq:contact_angle_def}, we have \begin{equation}\label{eq:b_dot_X_1}
  \mathbf{b}\cdot \bfu = v_{\cl}\, \mathbf{b}\cdot \bfn_{\cl} = v_{\cl}\cos\theta_{\cl}.
\end{equation}
Thus, only the $\bfn_{\cl}$-component contributes to the boundary term in \eqref{eq:div_surface_1}. Substituting \eqref{eq:b_dot_X_1} into \eqref{eq:div_surface_1} gives
\begin{equation}
\int_{S_t} \gamma(\Gamma) \operatorname{div}_S \bfu \, \ud\mathcal{H}^2
 = \int_{S_t} (\gamma \kappa V - \bfv_{\tng}\cdot \nabla_S \gamma) \, \ud\mathcal{H}^2 + 
 \int_{C_t} \gamma(\Gamma) v_{\cl}\cos\theta_{\cl}\,\ud\mathcal{H}^1.
 \label{eq:div_surface_2}
\end{equation}
Therefore, combining \eqref{eq:rate_surf_energy} and \eqref{eq:div_surface_2},
the rate of change of the surface free energy is
\begin{equation}
\frac{\ud}{\ud t} \calF_{\mathrm{surf}}
=\int_{S_t} (\gamma \kappa V -\bfv_{\tng}\cdot \nabla_S \gamma) \, \ud\mathcal{H}^2 + 
 \int_{C_t} \left(
 \gamma(\Gamma) \cos\theta_{\cl}+\gamma_{SL}-\gamma_{SG}\right)v_{\cl}\,\ud\mathcal{H}^1 
 + \int_{S_t} \bfJ \cdot \nabla_S\mu\, \ud\mathcal{H}^2.
\label{eq:Fdot_surf}
\end{equation}

\subsection{Rate of change of the electrostatic energy}
\label{sec:variations_electric}
Next, we compute the rate of change of the electrostatic energy $\calF_{\mathrm{elec}}$ defined in \eqref{eq:F_el}.
To facilitate the derivation, we decompose $\calF_{\mathrm{elec}}$ into two contributions, the bulk electrostatic energy and the surface electrostatic energy,
\begin{equation}
\calF_{\mathrm{elec}} = \calF_{\mathrm{elec-bulk}} + \calF_{\mathrm{elec-surf}},
\end{equation}
where
\begin{equation}
\calF_{\mathrm{elec-bulk}}
=
\sum_{\pm}\int_{\Omega_\pm(t)} \frac{\varepsilon_\pm}{2}|\nabla\phi_\pm|^2\,\dd^3x,
\quad
\calF_{\mathrm{elec-surf}}=\int_{S_t} q(\Gamma)\,\phi\,\dSsurf.
\label{eq:energy_elec-bulk-surf}
\end{equation}
We derive the rates of change of $\calF_{\mathrm{elec-bulk}}$ and $\calF_{\mathrm{elec-surf}}$ below.

For the bulk electrostatic energy 
$\calF_{\mathrm{elec-bulk}}$ defined in \eqref{eq:energy_elec-bulk-surf},
applying the Reynolds transport theorem to the two moving subdomains yields
\begin{equation}
\frac{\dd}{\dd t}\calF_{\mathrm{elec-bulk}}
=
\sum_{\pm}\int_{\Omega_\pm(t)} \varepsilon_\pm \,\nabla\phi_\pm\cdot\partial_t \nabla\phi_{\pm}\,\dd^3x
-\int_{S_t}\Bigl[\frac{\varepsilon}{2}|\nabla\phi|^2\Bigr]V\,\dSsurf.
\label{eq:bulk-1}
\end{equation}
Since $\nabla\cdot(\varepsilon\nabla\phi)=0$ in each phase, integration by parts gives
\begin{equation}
\frac{\dd}{\dd t}\calF_{\mathrm{elec-bulk}}
=
\int_{S_t}\left(-\varepsilon_+\partial_n\phi_+\,\partial_t\phi_{+}+\varepsilon_-\partial_n\phi_-\,\partial_t\phi_{-}
-\Bigl[\frac{\varepsilon}{2}|\nabla\phi|^2\Bigr]V\right)\dSsurf.
\label{eq:bulk-2}
\end{equation}

We introduce the following claim for the integrand in \eqref{eq:bulk-2}:
\begin{claim}
\label{claim:elec_bulk}
The integrand of \eqref{eq:bulk-2} satisfies
\begin{equation}
 -\varepsilon_+\partial_n\phi_+\,\partial_t\phi_{+}+\varepsilon_-\partial_n\phi_-\,\partial_t\phi_{-}
-\Bigl[\frac{\varepsilon}{2}|\nabla\phi|^2\Bigr]V =    [\bfTM]\bfn\cdot \bfu-q\,\DtS\phi,
\label{eq:claim}
\end{equation}
where $\bfTM$ denotes the Maxwell stress,
\begin{equation}
\bfTM :=\varepsilon\bfE \otimes \bfE  -
\frac12 \varepsilon|\bfE|^2 \bfI,
\label{eq:maxwell_stress}
\end{equation}
where $\bfE = -\nabla\phi$ represents the electric field and $\bfI$ is the identity tensor. The quantity $\DtS\phi$ represents the \emph{surface material derivative} of the electrostatic potential $\phi$ along the moving interface $S_t$, 
\begin{equation}
\DtS\phi:=\partial_t\phi+\bfu\cdot\nabla\phi,
\label{eq:surface_material_derivative}
\end{equation}
where $\nabla$ represents the full gradient in $\RR^3$. 
\end{claim}
A proof of Claim~\ref{claim:elec_bulk} is provided in Appendix~\ref{sec:proof_claim2}.

Combining Claim \ref{claim:elec_bulk} with \eqref{eq:bulk-2} immediately gives the rate of change of the bulk electrostatic energy
\begin{equation}
\frac{\dd}{\dd t}\calF_{\mathrm{elec-bulk}}
=
\int_{S_t}
\left(
[\bfTM]\bfn\cdot \bfu-q\,\DtS\phi
\right)
\dSsurf.
\label{eq:boxed-bulk}
\end{equation}
\begin{remark} 
Consider the moving surface $S_t$ with velocity field $\bfu$.
A material point $\bfx(t)\in S_t$ moves along with the surface according to the velocity, $\dot \bfx(t)=\bfu(\bfx(t),t)$.
The electrostatic potential $\phi$ evaluated at this material point, $\phi(\bfx(t),t)$, therefore satisfies
\[
\frac{\mathrm d}{\mathrm dt}\phi(\bfx(t),t)
=
\partial_t\phi(\bfx(t),t)+\dot \bfx(t)\cdot\nabla\phi(\bfx(t),t) = 
\partial_t\phi(\bfx(t),t)+\bfu\cdot\nabla\phi(\bfx(t),t).
\]
This yields the definition of the surface material derivative in \eqref{eq:surface_material_derivative}.
\end{remark}

For the surface electrostatic energy
$\calF_{\mathrm{elec-surf}}$ defined in \eqref{eq:energy_elec-bulk-surf},
the surface transport equation \eqref{eq:general_derivative_patch} yields its rate of change
\begin{equation}
\frac{\dd}{\dd t}\calF_{\mathrm{elec-surf}}
=
\int_{S_t}
\left(
\DtS(q(\Gamma)\phi)+q(\Gamma)\phi\,\divS \bfu
\right)
\dSsurf,
\label{eq:surf-1}
\end{equation}
where we have used the definition of the surface material derivative \eqref{eq:surface_material_derivative}.

Using the transport equation for surfactant concentration \eqref{eq:transport_surfactant} together with the expansion of the material derivative
$
\DtS( q(\Gamma)\phi)=q(\Gamma)\DtS\phi+\phi q'(\Gamma)\DtS\Gamma,
$
we obtain
\begin{align}
\frac{\dd}{\dd t}\calF_{\mathrm{elec-surf}}
&=
\int_{S_t}
\left(
q(\Gamma)\DtS\phi
+\phi q'(\Gamma)(-\Gamma\,\divS \bfu - \divS \bfJ)
+\phi q(\Gamma)\,\divS \bfu
\right)
\dSsurf \notag\\
&=
\int_{S_t}
\left(
q(\Gamma)\DtS\phi
+\bigl(q(\Gamma)-\Gamma q'(\Gamma)\bigr)\phi\,\divS \bfu - \phi q'(\Gamma)\divS\bfJ
\right)
\dSsurf.
\label{eq:surf-final}
\end{align}
Since we assume the linear relation $q(\Gamma) = A\Gamma$ (see \eqref{eq:q_linear_Gamma}), \eqref{eq:surf-final} reduces to
\begin{equation}
\frac{\dd}{\dd t}\calF_{\mathrm{elec-surf}}
=
\int_{S_t}
\left(q(\Gamma)\DtS\phi - A \phi \divS\bfJ\right)
\dSsurf.
\label{eq:surf_simplified}
\end{equation}
Substituting \eqref{eq:boxed-bulk} and \eqref{eq:surf_simplified} into the time derivative of the electrostatic energy \eqref{eq:F_el}, and using the no-flux boundary condition \eqref{eq:noflux_bc_Gamma} together with the surface divergence theorem, we obtain
\begin{equation}
\label{eq:dot_Felec}
\frac{\dd}{\dd t}\calF_{\mathrm{elec}} = 
\int_{S_t}
\left([\bfTM]\bfn\cdot \bfu + A \bfJ\cdot \nabla_S\phi \right)
\dSsurf.
\end{equation}

\subsection{The rate of change of the free energy and unbalanced forces identified}
\label{sec:unbalanced_forces}

Substituting \eqref{eq:Fdot_surf} and \eqref{eq:dot_Felec} into the total free energy \eqref{eq:F_total},
we obtain the rate of change of the total free energy $\calF$,
\begin{align}
\label{eq:Fdot}
 \frac{\dd \calF}{\dd t}&=
\int_{S_t} (\gamma \kappa V -\bfv_{tng}\cdot \nabla_S \gamma) \, \ud\mathcal{H}^2 + 
 \int_{C_t} \left(
 \gamma(\Gamma) \cos\theta_{\cl}+\gamma_{SL}-\gamma_{SG}\right)v_{\cl}\,\ud\mathcal{H}^1 \nonumber \\
& + \int_{S_t}
\left([\bfTM]\bfn\cdot \bfu + \bfJ\cdot \nabla_S(\mu + A\phi)\right)
\dSsurf.
\end{align}
To further simplify this expression, we define the total surface force
\begin{equation}
\bff_{\mathrm{surf}} = \kappa \gamma\bfn
      - \grads \gamma,
\label{eq:total_surface_stress}
\end{equation}
where the first term represents the normal capillary force due to surface tension, and the second term corresponds to the tangential Marangoni stress induced by surfactant concentration gradients.
Substituting \eqref{eq:total_surface_stress} into \eqref{eq:Fdot} yields
\begin{equation}
\label{eq:Fdot2}
 \frac{\dd \calF}{\dd t}=\mathcal{I}_{\mathrm{surf}} + \mathcal{I}_{\mathrm{\cl}}+\mathcal{I}_{\mathrm{diff}},
\end{equation}
where
\begin{equation}
    \mathcal{I}_{\mathrm{surf}} = \int_{S_t} (\bff_{\mathrm{surf}} + [\bfTM]\bfn) \cdot \bfu \ud\mathcal{H}^2, 
\label{eq:I_surf}
\end{equation}
\begin{equation}
        \mathcal{I}_{\mathrm{\cl}} =  \int_{C_t} \left(
 \gamma(\Gamma) \cos\theta_{\cl}+\gamma_{SL}-\gamma_{SG}\right)v_{\cl}\,\ud\mathcal{H}^1,
\label{eq:I_cl}
\end{equation}
and
\begin{equation}
    \mathcal{I}_{\mathrm{diff}} = \int_{S_t} \bfJ\cdot \nabla_S(\mu + A\phi) \, \ud\mathcal{H}^2.
\label{eq:I_diff}
\end{equation}

Equation \eqref{eq:Fdot2} shows that the total rate of change of the free energy consists of three contributions:
(1) The first integral $\mathcal{I}_{\mathrm{surf}}$ describes the contributions of the unbalanced normal and tangential forces at the interface arising from the surface and Maxwell stresses. This contribution is discussed in more detail below in this subsection.
(2) The second integral $\mathcal{I}_{\mathrm{\cl}}$ represents the contribution of the contact-line dynamics through the unbalanced Young's force $\gamma(\Gamma) \cos\theta_{\cl}+\gamma_{SL}-\gamma_{SG}$ multiplied by the contact-line velocity $v_{\cl}$. 
(3) The third integral $\mathcal{I}_{\mathrm{diff}}$ represents the contribution of surface diffusion through the unbalanced diffusion force in Fick's law multiplied by the diffusive flux $\bfJ$.
These unbalanced forces identified in \eqref{eq:Fdot2} will be determined through Onsager's variational principle by pairing them with appropriate dissipation functionals, as discussed in subsections \ref{sec:dissipation} and \ref{sec:rayleighian}.

Next, we discuss the unbalanced forces identified in $\mathcal{I}_{\mathrm{surf}}$ defined in \eqref{eq:I_surf}.
Using the definition of surface tension together with its relation to the chemical potential $\mu$ in \eqref{eq:surface_energy_chemical}, we express the surface tension gradient along the interface as
\begin{equation}
\nabla_S \gamma = -\Gamma \nabla_S \mu.
\end{equation}
Consequently, the total surface force $\mathbf{f}_{\mathrm{surf}}$ defined in \eqref{eq:total_surface_stress} can be written as
$
\mathbf{f}_{\mathrm{surf}} = \kappa \gamma \mathbf{n} + \Gamma \nabla_S \mu.
$
Therefore, the contribution of the surface force to the rate of change of the free energy is 
\begin{equation}
\int_{S_t} (\gamma \kappa V +\bfv_{tng}\cdot \Gamma\nabla_S \mu) \, \ud\mathcal{H}^2,
\end{equation}
where $V$ denotes the normal velocity of the interface and $\mathbf{v}_{\mathrm{tng}}$ is the tangential velocity (see \eqref{eq:velocity_decomp}). Physically, the first term represents the work done by capillary forces in the normal direction, and the second term accounts for the work associated with Marangoni-driven tangential motion induced by gradients in chemical potential.

Similarly, using \eqref{eq:maxwell-traction-identity} and \eqref{eq:maxwell-traction-identity2}, the contribution from the Maxwell traction can be decomposed as
\begin{equation}
[\bfTM\bfn]\cdot \bfu = V \, [\mathbf{T}_M]\mathbf{n} \cdot \mathbf{n} 
+ \mathbf{v}_{\mathrm{tng}} \cdot [\mathbf{T}_M]\mathbf{n}
=
V \left[\varepsilon (\partial_n \phi)^2 - \frac{\varepsilon}{2} |\nabla \phi|^2 \right]
+ \mathbf{v}_{\mathrm{tng}} \cdot \Gamma \nabla_S (A \phi),
\end{equation}
where the first term represents the normal electrostatic stress, while the second term describes the tangential electrostatic forcing along the interface.

Combining all contributions, the integral $\mathcal{I}_{\mathrm{surf}}$ in \eqref{eq:I_surf} becomes
\begin{equation}
\mathcal{I}_{\mathrm{surf}}=
\int_{S_t}
\left(V \left(\kappa \gamma+
\left[\varepsilon (\partial_n \phi)^2 - \frac{\varepsilon}{2} |\nabla \phi|^2 \right]
\right)
+ \mathbf{v}_{\mathrm{tng}} \cdot \Gamma \nabla_S \left( A \phi + \mu \right) \right)\ud \mathcal{H}^2,
\label{eq:rate_energy_change_fsurf_Tm}
\end{equation}
where the two terms in the integrand correspond to the combined contributions of the unbalanced normal and tangential stresses, respectively.
Physically, the normal component governs interface deformation through the interaction between capillary and electrostatic pressures, while the tangential component drives interfacial flows through gradients in chemical potential and electric potential.

\section{Dissipation functional and Onsager's principle}
\label{sec:dissipation_Rayleighian}
In this section, we introduce the dissipation functional $\calD$ within Onsager's variational framework to determine the unbalanced forces identified through the rate of change of the free energy \eqref{eq:Fdot}. The construction of the dissipation functional is presented in Subsection~\ref{sec:dissipation}. The Rayleighian and its minimization, together with the resulting coupled system, are then presented in Subsections \ref{sec:rayleighian} and \ref{sec:system}, respectively.

\subsection{Dissipation functional}
\label{sec:dissipation}
We introduce three dissipation terms corresponding to the three groups of unbalanced forces identified through the rate of change of the free energy \eqref{eq:Fdot2}. Specifically,
\begin{itemize}
\item To account for the unbalanced normal and tangential stresses at the interface ($\mathcal{I}_{\mathrm{surf}}$ in \eqref{eq:I_surf}),  we introduce the viscous dissipation $\calD_{\mathrm{visc}}$ for Newtonian fluids,
\begin{equation}
  \label{eq:D_visc}
  \calD_{\mathrm{visc}}[\bfu]
  = \int_{\Omega_{+}\cup\Omega_{-}}\mu_{\pm}\,|\bfD(\bfu)|^2\,\dd^3 x,
\end{equation}
where $|\bfD|^2 = \bfD:\bfD = \sum_{i,j}D_{ij}^2$ denotes the squared
Frobenius norm of the rate-of-strain tensor, and the velocity $\bfu$ satisfies the no-slip and no-penetration boundary condition
\begin{equation}
    \bfu = 0 \quad \mbox{ on } S_w.
\label{eq:no-slip}
\end{equation}
We note that Navier slip boundary conditions can also be used to alleviate the contact-line singularity \citep{qian2003molecular,ren2011contact}.

\item To pair with the unbalanced Young's force ($\mathcal{I}_{\mathrm{\cl}}$ in \eqref{eq:I_cl}), we introduce the dissipation associated with contact-line motion
\begin{equation}
    \label{eq:dissipation_cl}
\calD_{\cl}[v_{\cl}] = \frac{\zeta}{2}\int_{C_t}|v_{\cl}|^2~\dd \ell,
\end{equation}
where $\zeta$ is the contact-line friction coefficient.

\item To account for the unbalanced surface diffusion force ($\mathcal{I}_{\mathrm{diff}}$ in \eqref{eq:I_diff}), we introduce the surface diffusion dissipation $\mathcal{D}_{\mathrm{diff}}$ arising from tangential transport of surfactant along the interface,
\begin{equation}
\calD_{\mathrm{diff}}[\bfJ] = \int_{S_t}\frac{|\bfJ|^2}{2\mathcal{M}(\Gamma)}\ud\mathcal{H}^2,
\label{eq:dissi_surface}
\end{equation}
where $\mathcal{M}$ is the mobility coefficient that depends on the surfactant concentration $\Gamma$.
\end{itemize}

Overall, the total dissipation functional is defined as the sum of the three contributions,
\begin{equation}
  \label{eq:D_total}
  \calD
  = \calD_{\mathrm{visc}} 
     + \calD_{\cl} + \calD_{\mathrm{diff}}.
\end{equation}

\subsection{The Rayleighian and its minimization}
\label{sec:rayleighian}

At each fixed time $t$, and for a given state
$(S_t, \GamE_t, \phi_t)$, we define the \emph{Rayleighian} based on the rate of change of the free energy \eqref{eq:Fdot2} and the total dissipation functional \eqref{eq:D_total} as
\begin{align}
  \label{eq:R_def}
  \calR
  &= \frac{\dd}{\dd t}\calF
    + \calD
    - \int_{\Omega_+\cup \Omega_-} p\,\divv\bfu\,\dd^3 x \notag\\
  &=
   \int_{S_t} (\bff_{\mathrm{surf}} + [\bfTM]\bfn) \cdot \bfu \ud\mathcal{H}^2+
 \int_{C_t} \left(
 \gamma(\Gamma) \cos\theta_{\cl}+\gamma_{SL}-\gamma_{SG}\right)v_{\cl}\,\ud\mathcal{H}^1 \notag\\
& + \int_{S_t} \bfJ\cdot \nabla_S(\mu + A\phi) \, \ud\mathcal{H}^2
+\int_{S_t}\frac{|\bfJ|^2}{2\mathcal{M}(\Gamma)}\ud\mathcal{H}^2 \notag\\
& + \int_{\Omega_{+}\cup\Omega_{-}}\mu_{\pm}\,|\bfD(\bfu)|^2\,\dd^3 x + \frac{\zeta}{2}\int_{C_t}|v_{\cl}|^2~\ud\mathcal{H}^1 - \int_{\Omega_+\cup \Omega_-} p\,\divv\bfu\,\dd^3 x,
\end{align}
where the last term enforces the incompressibility constraint $\divv\bfu = 0$ via the Lagrange multiplier $p$.

The governing equations are obtained by minimizing the Rayleighian,
\begin{equation}
  \label{eq:min_R}
  (\bfJ, \bfu, p, v_{\cl})
  = \operatorname*{arg\,min}_{\widetilde\bfJ,\widetilde\bfu, \widetilde p,  \widetilde v_{\cl}}
    \calR[\widetilde\bfJ, \widetilde\bfu, \widetilde p, \widetilde v_{\cl}].
\end{equation}
That is, the dynamics of $(S_t, \GamE_t, \phi_t)$ for $t\in [0,T]$ can be viewed as a trajectory on a Hilbert manifold. At a fixed time $t$, for a given point $p_t =(S_t, \GamE_t, \phi_t)$, the tangent direction $v_t$ in the tangent space $T_{p_t}$ is parameterized by  $(\bfJ, \bfu, p, v_{cl})$, and the evolution direction is determined by minimization of the Rayleighian \citep{GaoLiu2021IFB}.

Taking the first variation of $\calR$ with respect to $\bfJ$ and setting it to zero yields the diffusive surfactant flux
\begin{equation}
    \bfJ = -\mathcal{M}(\Gamma)\grads(\mu + A\phi).
\label{eq:tangential_flux}
\end{equation}
Substituting \eqref{eq:tangential_flux} into \eqref{eq:transport_surfactant} leads to the complete transport equation
\begin{equation}
\partial_t \Gamma + \bfu \cdot \nabla_S \Gamma + \Gamma \operatorname{div}_S \bfu = \operatorname{div}_S \left(\mathcal{M}(\Gamma)\grads(\mu + A\phi)\right).
\label{eq:transport_surfactant_complete}
\end{equation}
Similarly, the equivalent geometric form of the transport equation \eqref{eq:surfactant_transport_main} becomes
\begin{equation}
\partial_t \Gamma + \operatorname{div}_S(\Gamma \bfv_{\mathrm{tng}}) +  \kappa \Gamma V
= \operatorname{div}_S \left(\mathcal{M}(\Gamma)\grads(\mu + A\phi)\right).
\label{eq:surfactant_transport_main_complete}
\end{equation}

Taking the first variation of $\calR$ with respect to $p$ and setting it to zero gives the incompressibility condition
\begin{equation}
    \nabla\cdot\bfu = 0 \quad \text{in } \Omega_\pm.
\label{eq:incomp_condition}
\end{equation}
Taking the first variation of $\calR$ with respect to $\bfu$ and setting it to zero yields the Stokes equations together with the interfacial stress balance incorporating capillary and Maxwell stresses
\begin{align}
  \label{eq:Stokes_bulk}
  -\grad p + \divv(2\mu_{\pm}\bfD(\bfu)) &= \mathbf{0}
  &&\text{in } \Omega_\pm,\\
  \label{eq:Stokes_interface}
  \bigl[-p\bfI + 2\mu_{\pm}\bfD(\bfu) - \bfTM\bigr]\bfn
  &= \gamma\kappa\bfn - \grads \gamma
  &&\text{on } S_t.
\end{align}
The no-slip and no-penetration boundary condition on the velocity $\bfu$ is enforced through \eqref{eq:no-slip}.

Finally, minimizing $\calR$ with respect to the contact-line velocity $v_{\cl}$ gives the governing equation for the contact-line dynamics,
\begin{equation}
\zeta v_{\cl} = -\gamma(\Gamma)\cos\theta_{\cl} -(\gamma_{SL}-\gamma_{SG}).
\label{eq:ctl_condition}
\end{equation}

\subsection{The coupled system}
\label{sec:system}

Combining equations \eqref{eq:laplace} -- \eqref{eq:continuity_phi_u}, \eqref{eq:noflux_bc_Gamma}, \eqref{eq:transport_surfactant_complete}, and 
\eqref{eq:incomp_condition} -- \eqref{eq:ctl_condition}, we obtain the complete coupled system consisting of bulk equations in each bulk phase, 
\begin{subequations}
\label{eq:system}
\begin{alignat}{2}
  \divv\bfu &= 0 
  &&\quad\text{in }\Omega_\pm(t),
  \label{eq:sys_incomp}\\
  -\grad p + \mu_{\pm}\Delta \bfu &= \mathbf{0}
  &&\quad\text{in }\Omega_\pm(t),
  \label{eq:sys_Stokes}\\
  \divv(\varepsilon\grad\phi) &= 0
  &&\quad\text{in }\Omega_\pm(t),
  \label{eq:sys_Poisson}
\end{alignat}
the surfactant transport equations along the moving interface $S_t$,
\begin{alignat}{2}
\partial_t \Gamma + \operatorname{div}_S(\Gamma\bfu) &= - \operatorname{div}_S \bfJ, \qquad \bfJ = -\mathcal{M}(\Gamma)\grads(\mu + A\phi) &&\quad\text{on }S_t,
  \label{eq:sys_transport}
\end{alignat}
the interface conditions at the liquid-air interface $S_t$,
\begin{alignat}{2}
  \bigl[-p\bfI+2\mu_{\pm}\bfD(\bfu)-\bfTM\bigr]\bfn
  &= \gamma\kappa\bfn-\grads \gamma
  &&\quad\text{on }S_t,
  \label{eq:sys_stress}\\
  [\phi] = 0,\quad [\varepsilon\partial_n\phi] &= q(\Gamma), 
  &&\quad\text{on }S_t,
  \label{eq:sys_phi_jump}\\
  V &= \bfu\cdot\bfn
  &&\quad\text{on }S_t,
  \label{eq:Vn_interface_bc}
\end{alignat}
the interface condition on the solid substrate,
\begin{alignat}{2}
  \bfu &= 0 &&\quad \text{on } S_w,
  \label{eq:no_slip_bc}
  \\
  \phi & = \phi_0(x,y) &&\quad \text{at } z = 0,
  \label{eq:phi0_interface}
\end{alignat}
and the contact-line conditions,
\begin{alignat}{2}
\zeta v_{\cl} = -\gamma(\Gamma)\cos\theta_{\cl} &-(\gamma_{SL}-\gamma_{SG}) && \quad \text{on } C_t,
\label{eq:contactLine_bc}\\
\bfJ\cdot \bfb = & 0 &&\quad \mbox{on } C_t,
\label{eq:noflux_bc_Gamma_sys}
\end{alignat}
where the Maxwell stress $\bfTM$ is defined in \eqref{eq:maxwell_stress}, 
the surface charge density $q(\GamE)$ is defined in \eqref{eq:q_linear_Gamma}, and the surface tension $\gamma(\GamE)$ and chemical potential $\mu(\GamE)$ are defined in \eqref{eq:surface_energy_chemical}. The definition of the surface energy density $e(\GamE)$ is given by \eqref{eq:surface_energy_density}.
Appropriate far-field boundary conditions for $\bfu$ and $\phi$ on $\partial\Omega$ also need to be imposed to complete the system.
\end{subequations}

A typical choice of the mobility function $\calM$ is
\begin{equation}
\calM(\GamE) = \frac{D}{e''(\Gamma)} =\frac{D\GamE(\Gamma_s-\GamE)}{kT\Gamma_s},
\label{eq:choice_M}
\end{equation}
where $\Gamma_s$ is the saturated surfactant concentration, and $D$ is a diffusion coefficient. Using \eqref{eq:choice_M} we obtain the diffusive flux in \eqref{eq:sys_transport} as 
\begin{equation}
\bfJ = -D\nabla_S \Gamma - A\calM(\Gamma)\nabla_S \phi.
\label{eq:J_mobility_version}
\end{equation}

The system \eqref{sec:system} describes the coupled dynamics of the velocity field $\bfu$, pressure $p$, the electrostatic potential $\phi$, surfactant concentration $\GamE$, and the moving interface $S_t$. More specifically, 
the Stokes equations \eqref{eq:sys_incomp}--\eqref{eq:sys_Stokes}, together with the traction boundary condition \eqref{eq:sys_stress} and the no-slip and no-penetration boundary conditions \eqref{eq:no_slip_bc},
govern $(\bfu, p)$ in each bulk phase. The Laplace equation \eqref{eq:sys_Poisson}, together with the interfacial jump conditions \eqref{eq:sys_phi_jump} and the Dirichlet boundary condition \eqref{eq:phi0_interface}, governs the dynamics of the electrostatic potential $\phi$. The interface $S_t$ evolves according to the surface normal velocity $V$ described by \eqref{eq:Vn_interface_bc}, and
the contact-line velocity condition \eqref{eq:contactLine_bc} determines the contact-line dynamics. Finally, the surfactant transport equation \eqref{eq:sys_transport}
governs the evolution of $\GamE$ on the free interface $S_t$.
We also note that the geometric setup of the problem implicitly assumes that the moving interface $S_t$ remains attached to the solid substrate $z = 0$, that is, $S_t \cap \{z = 0\} = C_t$.

\section{A reduced system with Rayleigh dissipation}
\label{sec:reduced_Rayleigh}
While the system \eqref{sec:system} describes the comprehensive interactions among bulk fluid dynamics, interfacial evolution, surfactant transport, and electrostatic effects, numerically solving the full Stokes system coupled with moving interfaces and electrostatic equations can be computationally expensive, even with tailored numerical approaches such as level-set methods, volume-of-fluid methods, and immersed boundary method \citep{osher1988fronts,karimi2023multiphase,peskin2002immersed}. For many applications, it is desirable to develop reduced models that retain the essential physical mechanisms while significantly reducing computational complexity. 

In this section, we introduce a simplified formulation by modifying the dissipation functional. The resulting reduced system, presented in Subsection~\ref{sec:Rayleigh_system}, preserves the energetic structure of the original formulation while providing a more tractable framework for analysis and computation. Furthermore, as an illustrative example, Subsection~\ref{sec:1D_graph} presents a one-dimensional model for a surfactant-laden sessile droplet represented by a graph, followed by numerical simulations of the resulting system (see Subsections~\ref{sec:algorithm} and \ref{sec:numerics}.)

\subsection{The coupled system with Rayleigh dissipation}
\label{sec:Rayleigh_system}

To obtain a reduced and computationally more tractable system from \eqref{sec:system}, we replace the viscous dissipation functional $\calD_{\mathrm{visc}}$ in \eqref{eq:D_visc} by the Rayleigh dissipation functional \citep{goldstein1950classical,GaoLiu2021IFB},
\begin{equation}
\label{eq:rayleighian_Rayleigh}
\calD_{\mathrm{Ray}}[V,\mathbf{v}_{\mathrm{tng}}] = 
\frac{\xi_n}{2}\int_{S_t} V^2 \dd \mathcal{H}^2 + \frac{\xi_{\tau}}{2}\int_{S_t}|\mathbf{v}_{\mathrm{tng}}|^2 \dd \mathcal{H}^2,
\end{equation}
where $\xi_n$ and $\xi_{\tau}$ are normal and tangential friction coefficients, respectively.
Then with the total dissipation functional defined as
\begin{equation}
\widetilde{\calD} = \calD_{\mathrm{Ray}} + \calD_{\cl} + \calD_{\mathrm{diff}},  
\end{equation}
where the definitions of $\calD_{\cl}$ and  $\calD_{\mathrm{diff}}$ are given in \eqref{eq:dissipation_cl} and \eqref{eq:dissi_surface}, and using the expression \eqref{eq:rate_energy_change_fsurf_Tm}, we rewrite the
\emph{Rayleighian} in \eqref{eq:R_def} as
\begin{align}
  \label{eq:R_def_Rayleigh}
  \widetilde{\calR}
  &= \frac{\dd}{\dd t}\calF
    + \widetilde{\calD}
    -\lambda\left(|\Omega_-(t)|-\mathrm{Vol}\right) \notag\\
  &=
  \int_{S_t}
\left(V \left(\kappa \gamma+
\left[\varepsilon (\partial_n \phi)^2 - \frac{\varepsilon}{2} |\nabla \phi|^2 \right]
\right)
+ \mathbf{v}_{\mathrm{tng}} \cdot \Gamma \nabla_S \left( A \phi + \mu \right) \right)\ud \mathcal{H}^2 \notag\\
& \quad +
 \int_{C_t} \left(
 \gamma(\Gamma) \cos\theta_{\cl}+\gamma_{SL}-\gamma_{SG}\right)v_{\cl}\,\ud\mathcal{H}^1 
  + \int_{S_t} \bfJ\cdot \nabla_S( A\phi+\mu) \, \ud\mathcal{H}^2
+\int_{S_t}\frac{|\bfJ|^2}{2\mathcal{M}(\Gamma)}\ud\mathcal{H}^2 \notag\\
& \quad +\frac{\xi_n}{2}\int_{S_t} V^2 \dd \mathcal{H}^2 + \frac{\xi_{\tau}}{2}\int_{S_t}|\mathbf{v}_{\mathrm{tng}}|^2 \dd \mathcal{H}^2 + \frac{\zeta}{2}\int_{C_t}|v_{\cl}|^2~\ud\mathcal{H}^1 
 - \lambda\left(|\Omega_-(t)|-\mathrm{Vol}\right).
\end{align}

The modified Rayleighian \eqref{eq:R_def_Rayleigh} incorporates contributions of surface energy and electrostatic energy to the rate of change of free energy $\tfrac{\dd}{\dd t}\calF$, and contributions from Rayleigh dissipation, contact-line friction, and surface diffusion dissipation to the total dissipation. The term $\lambda\left(|\Omega_-(t)|-\mathrm{Vol}\right)$ uses a Lagrange multiplier $\lambda$ to impose the conservation of volume, where $|\Omega_-(t)|$ represents the Lebesgue measure of the liquid domain, and $\mathrm{Vol}$ represents the prescribed volume of the liquid domain. We note that in the original formulation of Rayleighian \eqref{eq:R_def}, the liquid volume is preserved through the incompressibility constraint imposed through the pressure $p$ as the Lagrange multiplier.

Similar to the coupled system \eqref{eq:system} associated with the Rayleighian formulation \eqref{eq:R_def}, the governing equations associated with the modified Rayleighian are obtained by minimizing \eqref{eq:rayleighian_Rayleigh} with respect to the admissible variables,
\begin{equation}
  \label{eq:min_R_Ray}
  (\bfJ, V, \mathbf{v}_{\mathrm{tng}}, v_{\cl})
  = \operatorname*{arg\,min}_{\widetilde\bfJ,\widetilde{V}, \widetilde{\mathbf{v}}_{\mathrm{tng}},  \widetilde v_{\cl}}
    \widetilde{\calR}[\widetilde\bfJ, \widetilde{V}, \widetilde{\mathbf{v}}_{\mathrm{tng}}, \widetilde v_{\cl}].
\end{equation}
Minimizing $\widetilde{\calR}$ with respect to $V$ gives
\begin{equation}
    \xi_n V = -\left(\kappa \gamma+
\left[\varepsilon (\partial_n \phi)^2 - \frac{\varepsilon}{2} |\nabla \phi|^2 \right]
\right) + \lambda,
\label{eq:normal_velocity_Ray}
\end{equation}
where $\lambda$ is the Lagrange multiplier associated with volume preservation.
Similarly, minimizing $\widetilde{\calR}$  with respect to $\mathbf{v}_{\mathrm{tng}}$ leads to
\begin{equation}
    \xi_{\tau} \mathbf{v}_{\mathrm{tng}} = -\Gamma \nabla_S \left( A \phi + \mu \right).
\label{eq:tng_velocity_Ray}
\end{equation}
Taking the first variation of $\widetilde{\calR}$ with respect to $\lambda$ and setting it to zero gives the volume preservation condition
\begin{equation}
|\Omega_-(t)| = \mathrm{Vol}.
\label{eq:vol_preservation}
\end{equation}
Taking the first variation of $\widetilde{\calR}$ with respect to $\bfJ$ and $v_{\cl}$ leads to the tangential surfactant flux \eqref{eq:tangential_flux} and the contact-line condition \eqref{eq:ctl_condition}, respectively, which are identical to those in the system \eqref{eq:system}.

Combining equations  \eqref{eq:laplace} -- \eqref{eq:continuity_phi_u}, \eqref{eq:noflux_bc_Gamma}, \eqref{eq:transport_surfactant_complete}, \eqref{eq:normal_velocity_Ray} --\eqref{eq:vol_preservation}, we obtain the reduced coupled system. It consists of a single bulk equation for the electrostatic potential $\phi$, 
\begin{subequations}
\label{eq:system_ray}
\begin{alignat}{2}
  \divv(\varepsilon\grad\phi) &= 0
  &&\quad\text{in }\Omega_\pm(t),
  \label{eq:sys_Poisson_Ray}
\end{alignat}
the surfactant transport equation
\begin{alignat}{2}
\partial_t \Gamma + \operatorname{div}_S(\Gamma\bfu) &= - \operatorname{div}_S \bfJ, \qquad \bfJ = -\mathcal{M}(\Gamma)\grads(\mu + A\phi) &&\quad\text{on }S_t,
  \label{eq:sys_transport_Ray}
\end{alignat}
the liquid-air interface conditions
\begin{alignat}{2}
\xi_n V &= -\left(\kappa \gamma+
\left[\varepsilon (\partial_n \phi)^2 - \frac{\varepsilon}{2} |\nabla \phi|^2 \right]
\right) + \lambda,
 &&\quad\text{on }S_t,
\label{eq:sys_normal_Ray}\\
\xi_{\tau} \mathbf{v}_{\mathrm{tng}} &= -\Gamma \nabla_S \left( A \phi + \mu \right),
 &&\quad\text{on }S_t,
\label{eq:sys_tangent_Ray}\\
  [\phi] &= 0,\quad [\varepsilon\partial_n\phi] = q(\Gamma), 
  &&\quad\text{on }S_t,
  \label{eq:sys_phi_jump_Ray}\\
  \bfu &= V \bfn + \mathbf{v}_{\mathrm{tng}}
  &&\quad\text{on }S_t,
  \label{eq:Vn_interface_bc_Ray}
\end{alignat}
the substrate boundary condition
\begin{alignat}{2}
  \phi & = \phi_0(z) &&\quad \text{at } z = 0,
  \label{eq:phi0_interface_Ray}
\end{alignat}
the contact-line conditions
\begin{alignat}{2}
\zeta v_{\cl} &= -\gamma(\Gamma)\cos\theta_{\cl} -(\gamma_{SL}-\gamma_{SG}) && \quad \text{on } C_t 
\label{eq:contactLine_bc_Ray}\\
\bfJ\cdot \bfb &= 0 &&\quad \mbox{on } C_t,
\label{eq:noflux_bc_Gamma_Ray}
\end{alignat}
and the volume preservation constraint
\begin{alignat}{2}
|\Omega_-(t)| = \mathrm{Vol}.
  \label{eq:sys_vol_preservation}
\end{alignat}
\end{subequations}
The constitutive relations for the surface charge density $q(\Gamma)$, surface tension $\gamma(\Gamma)$, chemical potential $\mu(\Gamma)$, and mobility function $\mathcal{M}(\Gamma)$, are identical to those used in the system \eqref{eq:system}.

\subsection{A reduced one-dimensional system with graph representation}
\label{sec:1D_graph}
To illustrate the derived energetic variational framework, we consider a simplified setting in which a sessile droplet with a small (acute) contact angle $\theta_{\cl} < 90^{\circ}$ is placed on a one-dimensional solid substrate. In this case, the free surface $S_t$ can be represented as a graph through the height function $h$,
$
S_t = \{(x,h(x,t): x \in (a(t), b(t))\}, 
$
and the wetting region on the solid substrate is given by $S_w(t) = [a(t), b(t)]\subset\RR^1$.
The liquid domain is defined as
$
\Omega_-(t)=\{(x,z): a(t)<x<b(t),\ 0<z<h(x,t)\},
$
and the gas domain is given by 
$\Omega_+(t)=\Omega\setminus \overline{\Omega_-(t)}$. 
This representation allows us to reduce the original system defined over a moving interface to a PDE system on a one-dimensional substrate $z=0$, which significantly simplifies numerical computation.

Given the height function $h(x,t)$, it is useful to define the tangential and normal unit vectors of the liquid-air interface as
$
\bftau={(1,\partial_x h)}/{\sqrt{g}}$ and
$
\bfn={(-\partial_x h,1)}/{\sqrt{g}},
$
respectively, where $g = 1+(\partial_x h)^2$.
Then the velocity restricted on the interface $S_t$ can be written as
$
\bfu|_{S_t}=V_n\bfn+V_{\tau}\bftau,
$
and the Cartesian components of $\bfu$ are given by
\begin{equation}
\label{eq:u-components}
 u_1=\frac{-V_n\partial_x h+V_{\tau}}{\sqrt{g}},
 \qquad
 u_2=\frac{V_n+V_{\tau}\partial_x h}{\sqrt{g}}.
\end{equation}
Moreover, the evolution of the free surface satisfies the kinematic boundary condition
\begin{equation}
\label{eq:kinematic}
 \partial_t h=u_2-u_1 \partial_x h=\sqrt{g}V_n.
\end{equation}
The curvature of the free surface is given by
\begin{equation}
\label{eq:curvature}
 \kappa=-\frac{\partial_{xx} h}{(1+(\partial_x h)^2)^{3/2}}=-\frac{\partial_{xx} h}{g^{3/2}}.
\end{equation}

To rewrite the coupled system with Rayleigh dissipation \eqref{eq:system_ray} under the graph representation, we introduce pullback representations for the variables defined on the liquid-air interface $S_t$. Specifically, let $F(x,z,t)$ be an arbitrary scalar function defined in a neighborhood of $S_t$. Its restriction to the graph $z=h(x,t)$ is represented by the pullback function
\begin{equation}
\label{eq:pullback}
f(x,t):=F(x,h(x,t),t),\qquad x\in(a(t),b(t)).
\end{equation}
Then the surface gradient of the scalar function and the surface divergence of the velocity field can be written as
\begin{equation}
\label{eq:surface-ops}
 \grads F=\frac{\partial_x f}{\sqrt{g}}\,\bftau,
 \qquad
 \divS\bfu=\partial_sV_{\tau}+\kappa V_n
 =\frac{\partial_xV_{\tau}}{\sqrt{g}}+\kappa V_n,
\end{equation}
where $\partial_s=g^{-1/2}\partial_x$ denotes the arclength derivative.

We now define the pullback of the interfacial traces of $\phi$ and $\Gamma$ by
\begin{equation}
\label{eq:pullback1a}
\varphi(x,t):=\phi(x,h(x,t),t),\qquad 
c(x,t):=\Gamma(x,h(x,t),t),\qquad x\in(a(t),b(t)).
\end{equation}
The symbol $\phi$ is reserved for the bulk electrostatic potential in $\Omega_\pm(t)$, whereas $\varphi$ denotes its interfacial trace pulled back onto the interval $(a(t),b(t))$.
Similarly, $c$ denotes the pullback of the interfacial surfactant concentration $\Gamma$.

The height function satisfies the contact-point conditions
\begin{equation}
\label{eq:contact-points}
h(a(t),t)=h(b(t),t)=0.
\end{equation}
Differentiating these identities with respect to time yields
\begin{equation}
\label{eq:endpoint-kinematic}
\partial_t h(a(t),t)+a'(t)\partial_x h(a(t),t)=0,
\qquad
\partial_t h(b(t),t)+b'(t)\partial_x h(b(t),t)=0.
\end{equation}
The compatibility condition \eqref{eq:endpoint-kinematic} can also be derived from the velocity components in \eqref{eq:u-components} by setting $u_1 = v_{\cl}$ and $u_2 = 0$ at the contact points $x = a(t), b(t)$, where $v_{\cl}$ is the contact-line velocity. This condition is equivalent to \eqref{eq:contact_compatibility} in the general geometric setting.

Next, we rewrite the system \eqref{eq:system_ray} for the graph surface case.
\begin{subequations}
\label{eq:2D-system-graph}
In the bulk regions, the electrostatic potential still satisfies the Laplace equation 
\begin{equation}
-\divv(\varepsilon\grad\phi)=0 \qquad
\text{in }\Omega_\pm(t)
\label{eq:2D-system-graph-bulk}
\end{equation}
subject to the interfacial jump condition at the interface $z = h(x,t)$
\begin{equation}
[\phi]=0,
\qquad
[\varepsilon\partial_n\phi]=Ac,
\qquad \text{on }z=h(x,t).
\label{eq:2D-system-graph-jump}
\end{equation}
and the boundary condition on the solid substrate at $z = 0$
\begin{equation}
\phi(x,0,t)=\phi_0(x,t).
\label{eq:2D-system-graph-substratephi}
\end{equation}
Using the kinematic condition \eqref{eq:kinematic} and the definition of curvature \eqref{eq:curvature}, the normal force balance \eqref{eq:normal_velocity_Ray} becomes a scalar evolution equation for the height function $h(x,t)$,
\begin{equation}
\xi_n \partial_t h
=\gamma(c)\frac{\partial_{xx}h}{g}-
\sqrt{g}\mathcal{T}_M(x,t)+\lambda \sqrt{g}, \qquad x\in(a(t),b(t)),
\label{eq:2D-system-graph-normal}
\end{equation}
where the Maxwell stress contribution enters through the pull-back quantity $\mathcal{T}_M$ evaluated on the interface
\begin{equation}
\mathcal{T}_M(x,t) = \left[\varepsilon (\partial_n\phi)^2-\frac{\varepsilon}{2}|\grad\phi|^2\right]\Big|_{z=h(x,t)}.
\label{eq:2D-system_Maxwell}
\end{equation}
The tangential force balance \eqref{eq:sys_tangent_Ray} and the kinematic condition similarly reduce to
\begin{equation}
V_{\tau}=-\frac{c}{\xi_{\tau} \sqrt{g}}\,\partial_x(\mu(c)+A\varphi),\qquad 
V_n = \frac{\partial_t h}{\sqrt{g}},\qquad
x\in(a(t),b(t)),
\label{eq:2D-system-graph-tangent}
\end{equation}
where the expression of $V_{\tau}$ incorporates the combined effects of Marangoni stress and electrostatic forces on the tangential velocity.

It is useful to define the scalar pull-back of the diffusive surfactant flux $J$ as the tangential component of the diffusive flux $\bfJ$ along the interface, $J = \bfJ(x,h(x,t),t)\cdot \bftau$. Using the identity $\divS \bfJ = \tfrac{1}{\sqrt{g}}\partial_x J$,
the conservation of surfactant mass on an arbitrary surface patch $A_t \subset S_t$, given by \eqref{eq:surfactant_conservation}, pulls back to the integral balance
\begin{equation}
    \frac{\dd}{\dd t}\int_{\omega_t}c\sqrt{g}\,\dd x = -\int_{\omega_t} \frac{\partial_x J}{\sqrt{g}}\sqrt{g}\,\dd x = -\int_{\omega_t} \partial_x J\,\dd x,\notag
\label{eq:graph_surf_1}
\end{equation}
where $\omega_t$ denotes the projection of $A_t$ onto the substrate $S_w$.
Applying the surface transport theorem \eqref{eq:transport-graph} for graph surfaces (Lemma \ref{lemma:graph-surface} in Appendix~\ref{Appendix:surface_transport_graph}), we write
\begin{equation}
 \frac{\dd}{\dd t}\int_{\omega_t}c\sqrt{g}\,\dd x = \int_{\omega_t}\partial_t(c\sqrt{g})+\partial_x (c(V_{\tau}-V_n\partial_x h))\,\dd x,
\label{eq:graph_surf_2}\notag
\end{equation}
where we have used the velocity decomposition \eqref{eq:u-components}. 
Combining the above two identities and using the fact that $\omega_t$ is arbitrary, we obtain the pull-back form of the surfactant transport equation  \eqref{eq:sys_transport_Ray} as
\begin{equation}
\partial_t(c\sqrt{g})+\partial_x\bigl(c(V_{\tau}-V_n\partial_x h)\bigr)
= -\partial_x J, \qquad x\in(a(t),b(t)),
\label{eq:sys_transport_graph}
\end{equation}
where from \eqref{eq:sys_transport} the pull-back diffusive flux $J$ is given by
\begin{equation}
J = -\frac{\calM(c)}{\sqrt{g}}\,\partial_x(\mu(c)+A\varphi).
\nonumber
\end{equation}
In the special case when the mobility $\mathcal{M}$ takes the form of \eqref{eq:choice_M}, the diffusive flux $J$ becomes
\begin{equation}
    J = -\frac{1}{\sqrt{g}}\left(D\partial_x c + A\calM\partial_x\varphi\right).
\label{eq:2D-system-graph-transport}
\end{equation}

The no-flux boundary condition \eqref{eq:noflux_bc_Gamma_Ray} for the surfactant reduces to the boundary condition at the two contact points,
\begin{equation}
J = 0, \quad \mbox{ at } x = a(t), b(t).
\label{eq:2D-noflux_bc_Gamma_graph}
\end{equation}
Combining \eqref{eq:contact-points} and \eqref{eq:contactLine_bc_Ray}, we have the contact-line conditions
\begin{align}
h(a(t),t)&=h(b(t),t)=0,
\label{eq:2D-system-graph-contactpoint}\\
\zeta a'(t)&=\gamma(c(a(t),t)){\cos\theta_{\cl}^{\ell}}+(\gamma_{SL}-\gamma_{SG}),
\label{eq:2D-system-graph-CLleft}\\
\zeta b'(t)&=-\gamma(c(b(t),t)){\cos\theta_{\cl}^{r}}-(\gamma_{SL}-\gamma_{SG}),
\label{eq:2D-system-graph-CLright}
\end{align}
where $\cos\theta_{\cl}^{\ell}$ and $\cos\theta_{\cl}^{r}$ are the contact angles at the left and right contact points, $x = a(t)$ and $x = b(t)$, respectively. In the graph representation, we can also write $\cos\theta_{\cl}^{\ell} = g^{-1/2}(a(t),t)$ and $\cos\theta_{\cl}^{r} = g^{-1/2}(b(t),t)$.
Finally, conservation of liquid volume is expressed through the integral constraint on the height function,
\begin{align}
\int_{a(t)}^{b(t)} h(x,t)\,\dd x&=\mathrm{Vol}.
\label{eq:2D-system-graph-volume}
\end{align}
\end{subequations}
The coupled system \eqref{eq:2D-system-graph} provides the desired reduced framework consisting of a two-dimensional electrostatic problem coupled with a one-dimensional moving-boundary system for pullback variables.
Specifically, the bulk unknown remains $\phi(x,z,t)$ in the two-dimensional domain $\Omega_\pm(t)$, and the interfacial variables $h(x,t)$, $c(x,t)$, and $\varphi(x,t)$ are reduced to the one-dimensional interval $(a(t),b(t))$.

\subsection{\texorpdfstring{Algorithm for graph system \eqref{eq:2D-system-graph}}{Algorithm for the graph system}}
\label{sec:algorithm}
Next, we present a semi-discrete algorithm (Algorithm \ref{alg:1}) based on a first-order implicit-explicit (IMEX) time-stepping scheme for solving the graph system \eqref{eq:2D-system-graph}. Following the work by \cite{GaoLiu2021IFB}, the moving grid, contact-line locations, height function, and surfactant concentration are updated sequentially at each time step. 
The Arbitrary Lagrangian-Eulerian (ALE) method \citep{hirt1974arbitrary} is used to remap the variables associated with the moving domain. The evolution equation for the height function \eqref{eq:2D-system-graph-normal} is discretized using an IMEX algorithm, and the surfactant transport equation \eqref{eq:sys_transport_graph} is solved using a fully implicit backward Euler scheme with Newton's iteration.

\vspace{0.2in}
\newcounter{algorithmcounter}

\refstepcounter{algorithmcounter}
\noindent
\begin{minipage}{\textwidth}
\hrule
\vspace{0.5em}
\noindent\textbf{Algorithm~\thealgorithmcounter.}
IMEX algorithm for the graph system~(\ref{eq:2D-system-graph}).
\label{alg:1}
\vspace{0.5em}
\hrule
\vspace{0.75em}
\begin{algorithmic}[1]
\STATE Initialize $\phi^0, h^0, c^0, a^0$ and $b^0$ at time $t^0 = 0$.
\STATE \textbf{for} $n=0,1,\cdots$ \textbf{do}
\STATE \quad \quad Given $\phi^n, h^n, c^n, a^n$ and $b^n$ at time step $t^n$, update $\phi^{n+1}, h^{n+1}, c^{n+1}, a^{n+1}$ and $b^{n+1}$ at time step $t^{n+1}$ as follows:
\STATE \quad \quad 
Update $a^{n+1}$ and $b^{n+1}$ using the contact-line equations (\ref{eq:2D-system-graph-CLleft}) and (\ref{eq:2D-system-graph-CLright}) with the explicit forward Euler scheme.
\STATE \quad \quad  Perform a linear remapping of the variables associated with the moving domain from $x\in[a^n, b^n]$ to $\hat{x}\in[a^{n+1}, b^{n+1}]$, such that $\rho^n(x) = \hat{\rho}^n(\hat{x})$, where $\rho^n \in \{h^n, c^n, g^n, \mathcal{T}_{M}^n\}$.
\STATE \quad \quad  
Update the height function from $\hat{h}^n(\hat{x})$ to $h^{n+1}(\hat{x})$ for $\hat{x} \in (a^{n+1}, b^{n+1})$ using (\ref{eq:2D-system-graph-normal}), (\ref{eq:2D-system_Maxwell}), (\ref{eq:2D-system-graph-volume}), and (\ref{eq:2D-system-graph-contactpoint})  with the scheme
\begin{equation}
\begin{cases}
\displaystyle\xi_n \frac{h^{n+1}(\hat{x})-\hat{h}^n(\hat{x})}{\Delta t}
=\frac{\gamma(\hat{c}^n)}{\hat{g}^{n}(\hat{x})}\partial_{xx}h^{n+1}(\hat{x})-
\sqrt{\hat{g}^n}\hat{\mathcal{T}}^n_M(\hat{x})+\lambda^{n+1} \sqrt{\hat{g}^n}(\hat{x}), 
\\
\int_{a^{n+1}}^{b^{n+1}} h^{n+1}(\hat{x},t)\,\dd \hat{x}=\mathrm{Vol},\\
h^{n+1}(a^{n+1}) = h^{n+1}(b^{n+1}) = 0.
\end{cases}
\end{equation}
\STATE \quad \quad  Update the surfactant concentration from $\hat{c}^n(\hat{x})$ to $c^{n+1}(\hat{x})$ using (\ref{eq:2D-system-graph-tangent}) -- (\ref{eq:2D-noflux_bc_Gamma_graph}) with the scheme
\begin{equation}
\begin{cases}
\displaystyle
\frac{(c^{n+1}\sqrt{g^{n+1}})(\hat{x})-(\hat{c}^n\sqrt{\hat{g}^{n}})(\hat{x})}{\Delta t} + \partial_{\hat{x}}\left({c}^{n+1}({V}^{n+1}_{\tau}-V_n^{n+1}\partial_{\hat{x}}h^{n+1})\right)(\hat{x}) = -\partial_{\hat{x}}J^{n+1}(\hat{x}), \\
\displaystyle J^{n+1}(\hat{x}) = -\frac{{\mathcal{M}}^{n+1}}{\sqrt{{g}^{n+1}}}\partial_{\hat{x}}(\mu(c^{n+1})+A\hat{\varphi}^n),\\
J^{n+1}(a^{n+1}) = J^{n+1}(b^{n+1}) = 0.
\end{cases}
\end{equation}
\STATE \quad \quad Update the electrostatic potential from $\phi^n$ to $\phi^{n+1}$ by solving the system (\ref{eq:2D-system-graph-bulk}) -- (\ref{eq:2D-system-graph-substratephi})
\begin{equation}
\begin{cases}
    -\divv(\varepsilon_{\pm}\grad\phi^{n+1})=0
\qquad \text{in }\Omega_\pm^{n+1},\\
[\phi^{n+1}]=0,
\qquad
[\varepsilon_{\pm}\partial_n\phi^{n+1}]=Ac^{n+1}
\qquad \text{on }z=h^{n+1},\\
\phi^{n+1} = \phi^{n+1}_0 \qquad \mbox{at } z = 0,\\
\mbox{Far-field boundary condition on } \partial \Omega.
\end{cases}
\end{equation}
\end{algorithmic}
\hrule
\vspace{0.5em}
\end{minipage}

\subsection{Numerical results}
\label{sec:numerics}
Next, we present numerical simulations of the reduced graph system \eqref{eq:2D-system-graph} to demonstrate the energetic variational framework and investigate the coupled effects of surfactant transport and electrostatic forcing on droplet dynamics. The initial droplet profile $h(x,t)$ is chosen as a symmetric sessile droplet with prescribed volume, and the initial surfactant concentration $c(x,t)$ is assumed to be spatially uniform,
\begin{equation}
    \label{eq:ic}
h(x,0) = -2(x-0.5)(x+0.5), \qquad c(x,0) = 0.5.
\end{equation}
The numerical simulations are carried out using the IMEX scheme described in Algorithm \ref{alg:1}. The electrostatic problem for $\phi$ is computed in the bulk domain $\Omega_{\pm}$ within the fixed computational domain $\Omega = [-1,1]\times[0,1]$. The height function $h(x,t)$ and the surfactant concentration $c(x,t)$ are evolved on the moving interval $(a(t), b(t))$, where the initial support at $t=0$ is given by $(a(0),b(0)) = (-0.5,0.5)$.
The other system parameters are $\xi_{\tau} =\xi_n= 1$, $\zeta=2$, $\gamma_{\mathrm{SL}}-\gamma_{\mathrm{SG}}=-0.7$, $\varepsilon_+=5$, $\varepsilon_-=1$, $\Gamma_s = k = T = 1$, $\gamma_0 = 2$, and $D = 0.1$.

We focus on two representative configurations of the substrate potential: a homogeneous potential and an inhomogeneous potential. The first case serves as a baseline example and examines the influence of electrostatic forcing in a spatially symmetric setting, and the second case investigates how spatially varying electric fields induce droplet motion and redistribution of surfactant.

\begin{figure}
    \centering
    \includegraphics[width=0.45\linewidth]{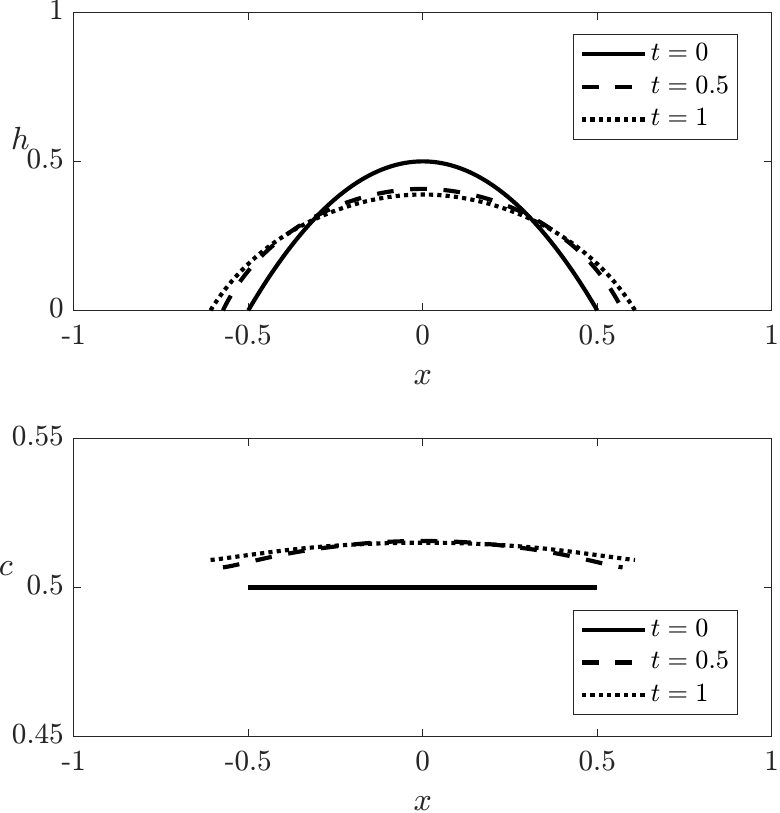}
    \includegraphics[width=0.48\linewidth]{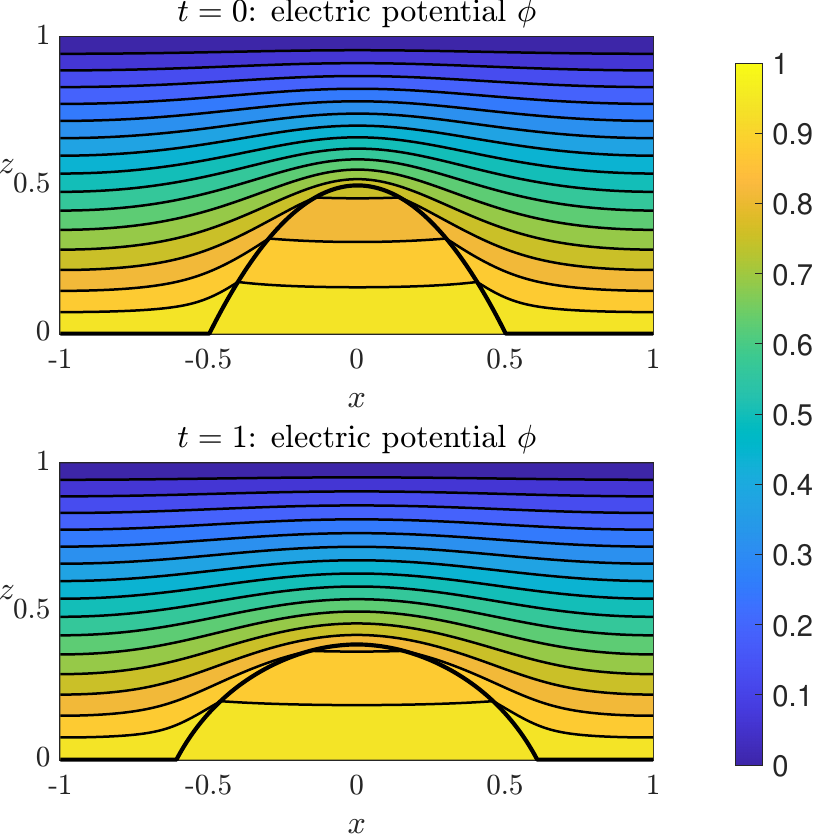}
    \caption{Droplet spreading under a homogeneous substrate electrostatic potential $\phi_0\equiv 1$. (Left) The height function $h(x,t)$ and pull-back surfactant concentration $c(x,t)$ at time $t = 0, 0.5, 1$; (Right) the electrostatic potential $\phi$ at time $t = 0$ and $t = 1$. The parameter $A = 0.1$.}
    \label{fig:phi_bottom_1}
\end{figure}

\begin{figure}
    \centering
    \includegraphics[width=0.45\linewidth]{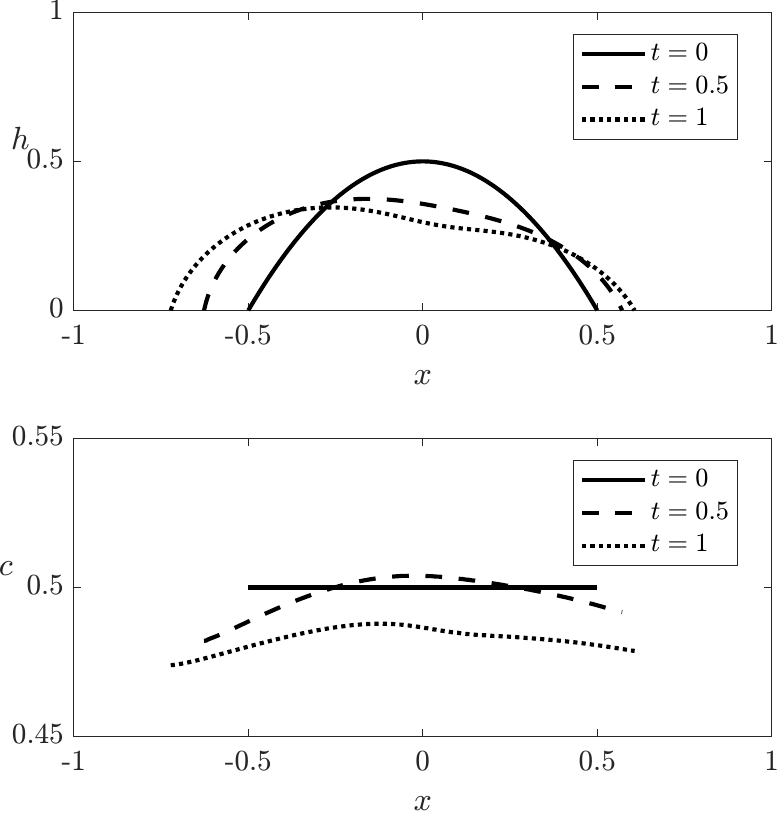}
    \includegraphics[width=0.48\linewidth]{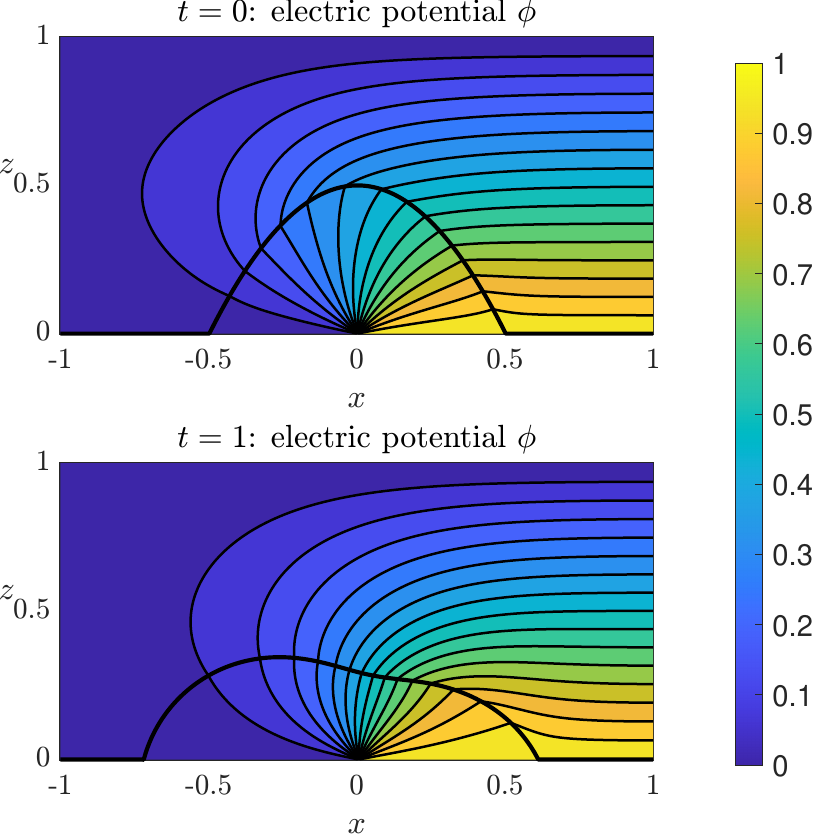}
    \caption{Droplet transport induced by the inhomogeneous electric field \eqref{eq:boundary_phi_inhomogeneous}, showing leftward droplet migration accompanied by surfactant redistribution. The other settings are identical to those in Figure~\ref{fig:phi_bottom_1}.}
    \label{fig:phi_left0_right1}
\end{figure}
Figure~\ref{fig:phi_bottom_1} presents the simulation of the sessile droplet evolving from the initial condition \eqref{eq:ic}, undergoing symmetric spreading on the solid substrate from time $t = 0$ to $t = 1$. A homogeneous boundary potential $\phi_0(x,0) \equiv 1$ is prescribed at the substrate $z=0$, and a Dirichlet boundary condition $\phi = 0$ is imposed at the top boundary of the computational domain, $z = 1$. Neumann boundary conditions $\phi_x = 0$ are imposed at the left and right boundaries, $x = -1$ and $x = 1$. The left panel of Figure~\ref{fig:phi_bottom_1} depicts the temporal evolution of the height function $h$ and surfactant concentration $c$ at times $t = 0, 0.5,$ and $1$, showing that the droplet spreads under the combined effects of motion by mean curvature, Marangoni forces, and electrostatic effects. The pull-back concentration $c$ increases over time as the geometric factor $\sqrt{g} = \sqrt{1+h_x^2}$ decreases during droplet spreading. 
The right panel of Figure~\ref{fig:phi_bottom_1} shows the electric potential distribution $\phi$ at the initial and final times, $t = 0$ and $t = 1$, respectively, both of which remain symmetric about $x = 0$. Throughout the evolution, the droplet preserves its symmetry while relaxing toward an equilibrium state. The electrostatic contribution modifies the balance between capillary and Maxwell stresses, contributing to the deformation of the free surface. Meanwhile, surfactant redistribution occurs due to the coupling between chemical potential gradients and electrostatic effects.

Next, we consider droplet dynamics induced by a spatially varying substrate potential, 
\begin{equation}
    \phi_0(x, 0) =
\begin{cases}
    0, \quad &-1 \le x < 0\\
    1, \quad &0\le x < 1
\end{cases}.
\label{eq:boundary_phi_inhomogeneous}
\end{equation}
Unlike the homogeneous case, the imposed substrate potential generates an electric field that breaks the left-right symmetry of the system.
Figure~\ref{fig:phi_left0_right1} illustrates the evolution of the droplet profile and the corresponding surfactant concentration under this inhomogeneous electric field, while all other settings are identical to those in Figure~\ref{fig:phi_bottom_1}. The asymmetry in electrostatic forcing creates tangential stresses along the interface and induces nonuniform surfactant transport. Specifically, since the surfactant concentration remains nearly uniform, the electrostatic contribution dominates the tangential force balance.
Consequently, the nonuniform substrate potential induces a negative tangential velocity $V_{\tau}$ over most of the droplet profile through \eqref{eq:2D-system-graph-tangent}. 
The droplet therefore exhibits directional motion toward the left driven by electrostatic forcing, accompanied by the dynamic redistribution of surfactant concentration (see the left panel). The right panel of Figure~\ref{fig:phi_left0_right1} shows the corresponding asymmetric electric potential $\phi$ in the computational domain generated by the inhomogeneous boundary potential.

\begin{figure}
    \centering
    \includegraphics[height=1.8in]{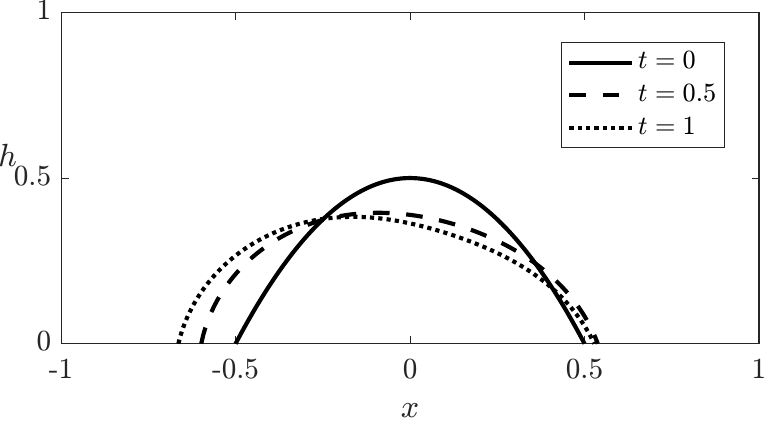}\quad
    \includegraphics[height=1.8in]{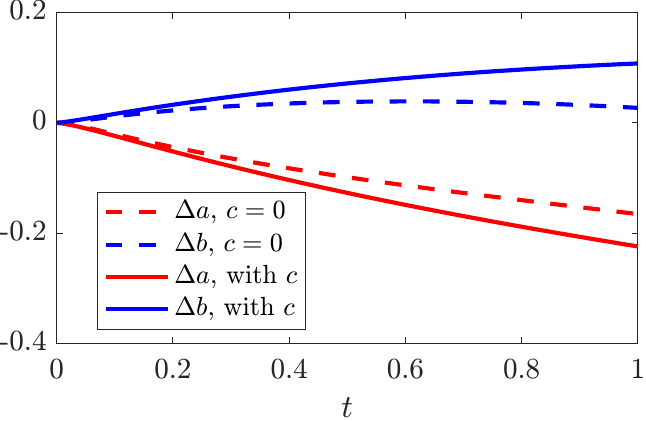}
    \caption{(Left) Droplet migration driven solely by Maxwell stresses. (Right) Comparison of the contact-line displacements $\Delta a(t) = a(t) - a(0)$ and $\Delta b(t) = b(t) - b(0)$ for a surfactant-laden droplet subject to coupled Maxwell stresses and Marangoni stresses (solid curves; Figure~\ref{fig:phi_left0_right1}) and a clean droplet driven solely by Maxwell stresses (dashed curves; Figure~\ref{fig:H_xab_c_comparison} (left)).}
    \label{fig:H_xab_c_comparison}
\end{figure}

The simulation in Figure~\ref{fig:phi_left0_right1} suggests that spatially varying electric fields can provide an effective mechanism for inducing droplet transport through the combined effects of Maxwell stresses and surfactant-mediated interfacial dynamics. To better illustrate the  respective roles of Maxwell stresses and surfactant transport in droplet migration, Figure~\ref{fig:H_xab_c_comparison} (left) presents the dynamics of a clean droplet without surfactant, obtained by setting the surfactant concentration $c(x,0) \equiv 0$ while taking the same initial height function $h(x,0)$ given by \eqref{eq:ic}. In this case, the governing system \eqref{eq:2D-system-graph} reduces to the electrostatic problem \eqref{eq:2D-system-graph-bulk} -- \eqref{eq:2D-system-graph-substratephi}, coupled with the height evolution equation \eqref{eq:2D-system-graph-normal} -- \eqref{eq:2D-system_Maxwell}, and the contact-line conditions and the volume constraint \eqref{eq:2D-system-graph-contactpoint} -- \eqref{eq:2D-system-graph-volume}. Consequently, the directional motion of the clean droplet is driven solely by Maxwell stresses. Comparing the snapshots in Figure~\ref{fig:phi_left0_right1} (top left) and Figure~\ref{fig:H_xab_c_comparison} (left) shows that, although both droplets exhibit a tendency to migrate toward the left, the clean droplet in Figure~\ref{fig:H_xab_c_comparison} (left) remains noticeably more rounded. This observation is consistent with the higher surface tension of the clean droplet compared to that of the surfactant-laden droplet. Moreover, Figure~\ref{fig:H_xab_c_comparison} (right) compares the contact-line displacements, $\Delta a(t) = a(t) - a(0)$ and $\Delta b(t) = b(t) - b(0)$. The results indicate that the presence of surfactant promotes droplet spreading during migration by reducing the overall surface tension, while the surfactant concentration gradient generates an additional Marangoni force that enhances directional transport. As a result of these combined effects, the advancing edge of the surfactant-laden droplet advances further to the left than that of the clean droplet.

Finally, we investigate the influence of the parameter $A$ in the constitutive relation \eqref{eq:q_linear_Gamma} on the electrohydrodynamics of surfactant-laden droplets. Figure~\ref{fig:HC_A_comparison} compares the evolution of both the height function $h(x,t)$ and the surfactant concentration $c$ for negatively charged surfactant $(A = -5)$ and positively charged surfactant $(A = 5)$, showing that changing the sign of $A$ can lead to either droplet wetting or dewetting in the case of symmetric droplet evolution. 
The parameter $A$ characterizes the dependence of the interfacial charge density on the surfactant concentration through the linear relation \eqref{eq:q_linear_Gamma} and therefore determines how the electric field couples to surfactant transport. Specifically, the sign of $A$ distinguishes different types of surfactants. When $A < 0$, the surfactant molecules carry a net negative charge (anionic surfactant), while the case $A > 0$ corresponds to positively charged (cationic) surfactants. We keep all the other settings identical to those in Figure~\ref{fig:phi_bottom_1}.

\begin{figure}
    \centering
    \includegraphics[width=0.48\linewidth]{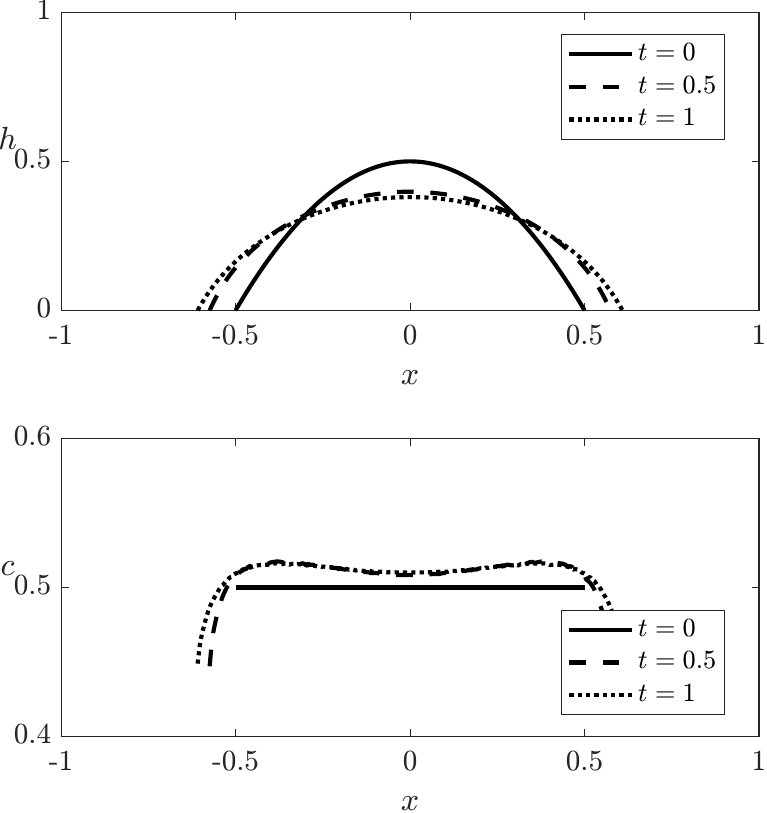}
    \includegraphics[width=0.48\linewidth]{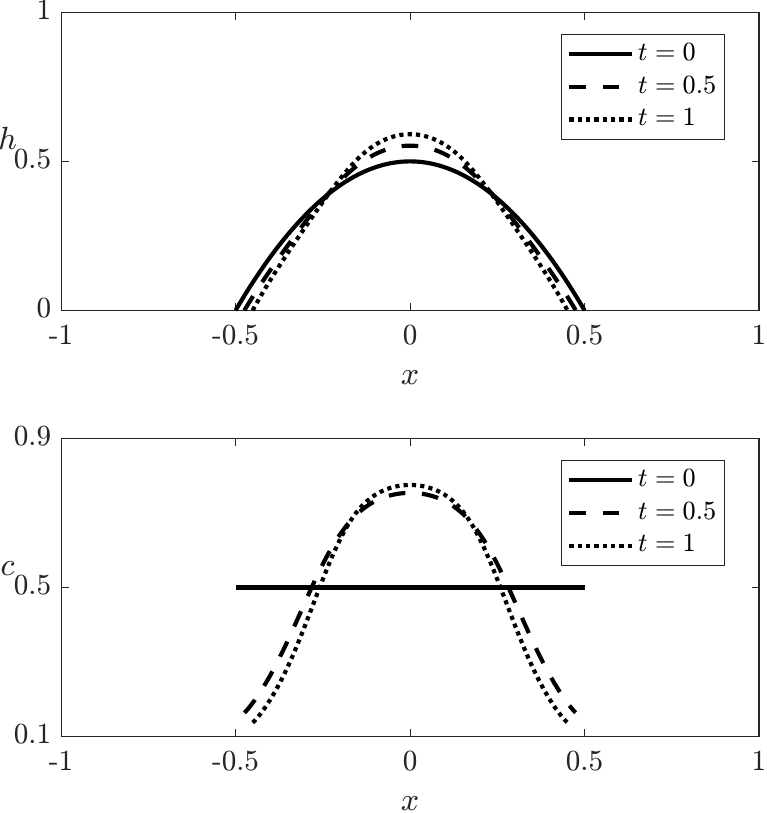}
    \caption{Comparison of the evolution of the height function $h(x,t)$ and the surfactant concentration $c(x,t)$ for (left) negatively charged surfactant $(A = -5)$ and (right) positively charged surfactant $(A = 5)$. The other parameters are identical to those in Figure~\ref{fig:phi_bottom_1}.}
    \label{fig:HC_A_comparison}
\end{figure}
The plots in Figure~\ref{fig:HC_A_comparison} show that, in both cases, the interfacial surfactant concentration $c$ is redistributed away from the contact lines of the droplet. However, the resulting droplet dynamics and surfactant distribution differ substantially because the Maxwell stress contribution to the tangential velocity in \eqref{eq:2D-system-graph-tangent} depends on the sign of the parameter $A$. Specifically, for $A = 5$, the surfactant concentration develops larger local spatial variations than the profile shown in Figure~\ref{fig:phi_bottom_1}, with more surfactant accumulating near the peak of the droplet. The associated surface tension gradient generates a Marangoni stress that acts against further surfactant accumulation in this region. In contrast, for $A = -5$, the surfactant concentration remains more uniformly distributed along the interface than in the positively charged case and exhibits a convex profile near the droplet peak, resulting in weaker surface tension gradients and consequently weaker Marangoni effects. The combined effects of Maxwell stresses and Marangoni stresses lead to rapid droplet spreading in the negatively charged case, while the positively charged surfactant case exhibits droplet dewetting.
These results demonstrate that the interplay between Maxwell and Marangoni stresses, and the resulting droplet dynamics, is strongly influenced by the constitutive relation \eqref{eq:q_linear_Gamma} and the parameter $A$.

\section{Concluding remarks}
\label{sec:conclusion}
The energetic variational framework developed in this work provides a unified formulation for droplet dynamics coupled with surfactant transport and electrostatic effects. The governing equations \eqref{eq:system} are derived directly from the interplay between free energy variations and dissipation mechanisms, and the reduced formulation \eqref{eq:system_ray} obtained through Rayleigh dissipation further provides a computationally tractable model while keeping the dominant physical effects. The graph representation of the model \eqref{eq:system_ray} further reduces the problem to a coupled system \eqref{eq:2D-system-graph} for the height function and the surfactant concentration defined on a one-dimensional domain.

The graph system \eqref{eq:2D-system-graph} shares similarities with lubrication-type models for thin films and surfactant-mediated electrohydrodynamics of droplets \citep{chu2023electrohydrodynamics}. Under the lubrication approximation, the moving interface is represented by a height function together with a long-wave assumption, resulting in reduced coupled equations for the height function and surfactant concentration coupled through a depth-averaged velocity. The proposed graph formulation \eqref{eq:2D-system-graph} may be viewed as a more geometrically general counterpart that does not rely on small-slope assumptions and can therefore be applied to sessile droplets with large contact angles. Moreover, classical lubrication models for droplet dynamics typically require regularization mechanisms to resolve contact-line singularities, such as precursor-film models with disjoining pressure or slip boundary conditions \citep{RevModPhys.57.827}. In contrast, the energetic formulation discussed in the present work incorporates contact-line dynamics through a dissipation mechanism associated with contact-line friction and avoids the introduction of additional regularization parameters.

Moreover, in many situations, sharp-interface models may be formally recovered as the sharp-interface limit of diffuse-interface models as the diffuse-interface thickness tends to zero \citep{qian2006variational}. It would be of interest to investigate whether the present moving contact-line law and associated Rayleighian structure can be derived systematically as the sharp-interface limit of an energetic variational phase-field model.

The present work adopts the linear relation $q(\Gamma) = A\Gamma$ (see \eqref{eq:q_linear_Gamma}) to model the linear dependence of the surface charge density on the surfactant concentration. This constitutive relation assumes that surfactant molecules carry electric charge, thereby directly coupling charge transport and surfactant transport. Other electrostatic models may also be incorporated within the present framework, such as the leaky dielectric model widely used in electrohydrodynamic systems \citep{Saville1997,papageorgiou2019film,nganguia2019effects}.

The present formulation is also restricted to the Stokes regime and neglects inertia effects. While this assumption is appropriate for many microfluidic applications involving slowly evolving droplets, inertia may become important for rapid transport processes and droplet coalescence.
Recent work by \cite{wray2022electrostatic} studied the electrostatic control in thin-film flows governed by the Navier-Stokes equations. Extending the current framework to incorporate inertia effects would be an interesting direction for future work.

From a numerical perspective, the electrostatic potential satisfies a Laplace equation in a moving domain, and efficient solution techniques become more important for large-scale simulations. For two-dimensional problems, for example, boundary integral methods \citep{sorgentone20193d} may provide an attractive alternative to bulk discretization approaches. Moreover, the development of structure-preserving numerical schemes that inherit the energy-dissipation relation associated with the Rayleighian \eqref{eq:R_def} at the discrete level would also be beneficial for large-scale long-time simulations.

\section*{Acknowledgments}
H. Ji's work is supported by NSF DMS-2309774. 


\section*{Data availability statement}
Datasets and code generated during the current study are available in the GitHub repository \cite{EHD-code}.

\appendix
\section{Proof of Claim \ref{claim:elec_bulk}}
\label{sec:proof_claim2}
\begin{proof}
Using the velocity decomposition \eqref{eq:velocity_decomp}, we rewrite the surface material derivative of the electric potential $\phi$ as
\[
\DtS\phi=\partial_t\phi+\bfu\cdot\nabla\phi
=\partial_t\phi+V\,\bfn\cdot\nabla\phi+\vtan\cdot\nabla_S\phi.
\]
Since $\partial_{\bfn}\phi:=\bfn\cdot\nabla\phi$, we obtain
\begin{equation}
\DtS\phi=\partial_t\phi+V\partial_{\bfn}\phi+\vtan\cdot\nabla_S\phi.
\label{eq:surf_material_derivative}
\end{equation}
Applying \eqref{eq:surf_material_derivative} to the electric potential $\phi_{\pm}$ on both sides of the interface $S_t$ leads to
\begin{equation}
\DtS\phi
=
\partial_t\phi_{\pm}+ \bfu\cdot\nabla\phi_\pm
=
\partial_t\phi_{\pm}+V \partial_{\bfn}\phi_\pm + \bfv_{\mathrm{tng}}\cdot\nabla_S\phi,
\label{eq:surfaceMaterialDerivative}
\end{equation}
which implies
\begin{equation}
\label{eq:phi_t_total}
\partial_t\phi_{\pm}
=
\DtS\phi-V \partial_{\bfn}\phi_\pm -\bfv_{\mathrm{tng}} \cdot\nabla_S\phi.
\end{equation}
Substituting \eqref{eq:phi_t_total} into the left-hand side of \eqref{eq:claim}, together with the jump condition \eqref{eq:jump_q}, yields
\begin{align}
\mathrm{LHS} &=
-\varepsilon_+\partial_n\phi_+\bigl(\DtS\phi-V\partial_{\bfn}\phi_+ -\bfv_{\mathrm{tng}}\cdot\nabla_S\phi\bigr)
\notag\\
&\qquad\qquad
+\varepsilon_-\partial_n\phi_-\bigl(\DtS\phi-V \partial_{\bfn}\phi_- -\bfv_{\mathrm{tng}}\cdot\nabla_S\phi \bigr)
-\Bigl[\frac{\varepsilon}{2}|\nabla\phi|^2\Bigr]V \notag\\
&=-q\,\DtS\phi
+
q\,\bfv_{\mathrm{tng}}\cdot\nabla_S\phi
+
\Bigl[\varepsilon(\partial_n\phi)^2-\frac{\varepsilon}{2}|\nabla\phi|^2\Bigr]V.
\end{align}

Next, we decompose the electric field $\bfE$ into its normal and tangential components,
\[
\bfE_\pm = \bfE_{\tau} + E_{n,\pm}\bfn,
\qquad
\bfE_{\tau}=-\nabla_S\phi,
\qquad
E_{n,\pm}=-\partial_n\phi_\pm.
\]
Using the definition of Maxwell stress \eqref{eq:maxwell_stress}, we write the Maxwell traction as
\begin{equation}
\bfT_{\mathrm{M},\pm}\bfn = \left( \varepsilon_\pm\bfE_\pm \otimes \bfE_\pm  -
\frac12 \varepsilon_\pm|\bfE_\pm|^2 \bfI \right)\bfn
=
\varepsilon_\pm
\left(
E_{n,\pm}\bfE_\pm-
\frac12|\bfE_\pm|^2\bfn
\right).
\label{eq:maxwell_traction}
\end{equation}
Using the velocity decomposition \eqref{eq:velocity_decomp},
we compute the normal and tangential contributions of the jump in Maxwell traction across the interface and express them in terms of the electric potential $\phi$,
\begin{equation}
[\bfTM]\bfn\cdot \bfn
=
\Bigl[\varepsilon E_n^2-\frac12\varepsilon|\bfE|^2\Bigr]
=
\Bigl[\varepsilon(\partial_n\phi)^2-\frac{\varepsilon}{2}|\nabla\phi|^2\Bigr],
\label{eq:maxwell-traction-identity}
\end{equation}
and
\begin{equation}
[\bfTM]\bfn\cdot \bfv_{\mathrm{tng}}
=
[\varepsilon E_n]\,\bfE_\tau\cdot 
\bfv_{\mathrm{tng}}
=
q\,\bfv_{\mathrm{tng}}\cdot\nabla_S\phi ,
\label{eq:maxwell-traction-identity2}
\end{equation}
where we have used $[\varepsilon E_n]=-[\varepsilon\partial_n\phi]=-q$ and $\bfE_\tau=-\nabla_S\phi$.
Combining \eqref{eq:maxwell-traction-identity} and \eqref{eq:maxwell-traction-identity2}, we rewrite the right-hand side of \eqref{eq:claim} as
\begin{align}
    \mbox{RHS} &= [\bfTM]\bfn\cdot \bfn V + [\bfTM]\bfn\cdot \bfv_{\mathrm{tng}} -q\,\DtS\phi \notag\\
    &=\Bigl[\varepsilon(\partial_n\phi)^2-\frac{\varepsilon}{2}|\nabla\phi|^2\Bigr]V + q\,\bfv_{\mathrm{tng}}\cdot\nabla_S\phi
-q\,\DtS\phi = \mbox{LHS}.
\end{align}
This completes the proof of Claim \ref{claim:elec_bulk}.
\end{proof}

\section{Surface transport theorem}
\label{sec:appendix_surf}
In this appendix, we present the surface transport theorem in Lemma~\ref{ref:lemma_surface} for a general moving interface $S_t$. We then rewrite the theorem in graph coordinates for the special case in which the interface $S_t$ admits a graph representation, $S_t = \{(x,y,h(x,y,t)) : (x,y)\in S_w\subset \RR^2\}$; see Lemma~\ref{lemma:graph-surface} in Appendix~\ref{Appendix:surface_transport_graph}.

\begin{lemma}[Surface transport theorem]
\label{ref:lemma_surface}
Let $S_t \subset \mathbb{R}^3$ be a smooth evolving surface with velocity field $X$ and outward normal unit vector $\bfn$, and let $f(\cdot,t)$ be a scalar function defined in a neighborhood of $S_t$.
Let $A_t \subset S_t$ be a material surface patch moving with the flow.
Then
\begin{equation}
\frac{d}{dt}\int_{A_t} f\, d\mathcal{H}^2
= \int_{A_t} \left(\partial_t f + X\cdot \nabla f  + f\, \operatorname{div}_S X\right)\, d\mathcal{H}^2.
\label{eq:general_derivative_patch}
\end{equation}
\end{lemma}

If, in addition, $f(\cdot, t)$ also satisfies $\bfn\cdot \nabla f = 0$, then we refer to $f$ as an \emph{insoluble} quantity, and equation \eqref{eq:general_derivative_patch} reduces to
\begin{equation}
\frac{d}{dt}\int_{A_t} f\, d\mathcal{H}^2
= \int_{A_t} \left(\partial_t f + X\cdot \nabla_S f  + f\, \operatorname{div}_S X\right)\, d\mathcal{H}^2.
\label{eq:insoluble_derivative_patch}
\end{equation}
\begin{remark}
The two forms of the surface transport theorem are used in different contexts throughout this paper:
\begin{itemize}
\item The form in \eqref{eq:general_derivative_patch} applies to fields defined in $\RR^3$, such as the electric potential $\phi$.
\item The form in \eqref{eq:insoluble_derivative_patch} applies to quantities constrained to the moving interface $S_t$, such as the insoluble surfactant concentration $\Gamma$. In this case, $\Gamma$ is locally extended to a neighborhood of $S_t$ so that the condition $\bfn\cdot \nabla \Gamma = 0$ is satisfied \citep{Stone1990}.
\end{itemize}
\end{remark}

\begin{proof}
Let $\Phi_t:S_0 \to S_t$ be the flow map generated by the velocity field $X$, and let 
\[
J_s(x,t) = \det\big(D\Phi_t(x)|_{T_x S_0}\big)
\]
denote the tangential Jacobian, which measures local change in surface area.

We rewrite the integral on the reference surface as
\begin{equation}
I(t) = \int_{A_t} f\, d\mathcal{H}^2
= \int_{A_0} f(\Phi_t(x),t)\, J_s(x,t)\, d\mathcal{H}^2(x).
\end{equation}
Differentiating with respect to time yields
\begin{align}
\frac{d}{dt} I(t)
&= \int_{A_0}
\left[
\frac{d}{dt} f(\Phi_t(x),t)\, J_s(x,t)
+ f(\Phi_t(x),t)\, \frac{\partial J_s}{\partial t}(x,t)
\right] d\mathcal{H}^2(x).
\end{align}
Applying the chain rule along trajectories gives
\[
\frac{d}{dt} f(\Phi_t(x),t)
= \partial_t f + X \cdot \nabla f,
\]
and the Jacobi's formula
\begin{equation}
\frac{1}{J_s}\frac{\partial J_s}{\partial t} = \operatorname{div}_S X,
\end{equation}
we obtain
\begin{align}
\frac{d}{dt} I(t)
&= \int_{A_t}
\left[
\partial_t f + X \cdot \nabla f + f\, \operatorname{div}_S X
\right] d\mathcal{H}^2.
\end{align}

If $f$ is insoluble and defined only on the surface $S_t$, then only its tangential derivative contributes to the integral, and we have
\[
X \cdot \nabla f = X \cdot \nabla_S f.
\]
This completes the proof.
\end{proof}

\section{Surface transport theorem for a graph surface}
\label{Appendix:surface_transport_graph}
We consider an evolving surface in $\RR^3$ represented as a graph,
\[
S_t = \{(x,y,h(x,y,t)) : (x,y)\in S_w\subset \RR^2\},
\]
where $h:S_w\times [0,T]\to\RR$ is smooth. The surface moves with the ambient velocity field 
$
\bfv=(v_x,v_y,v_z).
$
We write
$
\mathbf r(x,y,t)=(x,y,h(x,y,t))
$
for the graph parameterization.
The tangent vectors of the graph are
$
\mathbf r_x=(1,0,\partial_xh)$ and
$\mathbf r_y=(0,1,\partial_y h)$.
Hence we have
$
\mathbf r_x\times \mathbf r_y = (-\partial_x h,-\partial_y h,1),
$
and the area element on $S_t$ is
\begin{equation}\label{eq:area-element}
\dd\mathcal H^2 = \sqrt{g}\,\dd x\dd y,
\qquad \mbox{where }
g=|\mathbf r_x\times \mathbf r_y|^2 = {1+|\nabla h|^2}.
\end{equation}
The corresponding unit normal is
\begin{equation}\label{eq:normal}
\bfn = \frac{(-\partial_x h,-\partial_y h,1)}{\sqrt{1+|\nabla h|^2}}.
\end{equation}

Let $\phi(x,y,z,t)$ be a scalar function defined in a neighborhood of $S_t$. Its restriction to the graph is given by the pullback
\begin{equation}\label{eq:pullback_appdix}
\varphi(x,y,t):=\phi(x,y,h(x,y,t),t).
\end{equation}
By the chain rule,
\begin{equation}\label{eq:chain-rule}
\partial_t\varphi = \partial_t\phi + \partial_t h\,\partial_z\phi,
\qquad
\partial_x\varphi = \partial_x\phi + h_x\,\partial_z\phi,
\qquad
\partial_y\varphi = \partial_y\phi + h_y\,\partial_z\phi,
\end{equation}
where all derivatives of $\phi$ on the right-hand side are evaluated at $(x,y,h(x,y,t),t)$.

Now let $A_t\subset S_t$ be a material surface patch. Its projection onto the $(x,y)$-plane is
\[
\omega_t:=\{(x,y)\in S_w:(x,y,h(x,y,t))\in A_t\}.
\]
Since $A_t$ is material, the projected domain $\omega_t$ moves with the planar velocity field
\begin{equation}\label{eq:w-def}
\bfw=(u,v),
\qquad
u(x,y,t)=v_x(x,y,h(x,y,t),t),
\qquad
v(x,y,t)=v_y(x,y,h(x,y,t),t).
\end{equation}
To verify this, let $(x(t),y(t),h(x(t),y(t),t))$ be a material point on the surface. Differentiating gives
\[
\frac{\dd}{\dd t}\mathbf r(x(t),y(t),t)
=
\bigl(x'(t),y'(t),h_t+x'(t)h_x+y'(t)h_y\bigr).
\]
Since this equals $\bfv=(v_x,v_y,v_z)$ along the surface, we obtain
\begin{equation}\label{eq:kinematic-system}
x'(t)=u=v_x,
\qquad
y'(t)=v=v_y,
\qquad
\partial_t h+u \partial_x h+v \partial_y h=v_z.
\end{equation}
The last identity is the graph kinematic condition governing the evolution of the height function $h$.

\begin{remark}[{A useful identity for the pullback material derivative}]
Define the planar material derivative by
\begin{equation}\label{eq:Dt-def}
D_t:=\partial_t+u\partial_x+v\partial_y.
\end{equation}
Using \eqref{eq:chain-rule} and \eqref{eq:kinematic-system}, we compute
\begin{align}
D_t\varphi
&= \partial_t\varphi + u\partial_x\varphi + v\partial_y\varphi \notag\\
&= \bigl(\partial_t\phi+h_t\partial_z\phi\bigr)
+u\bigl(\partial_x\phi+h_x\partial_z\phi\bigr)
+v\bigl(\partial_y\phi+h_y\partial_z\phi\bigr) \notag\\
&= \partial_t\phi + u\partial_x\phi + v\partial_y\phi
+\bigl(h_t+uh_x+vh_y\bigr)\partial_z\phi \notag\\
&= \partial_t\phi + v_x\partial_x\phi + v_y\partial_y\phi + v_z\partial_z\phi \notag\\
&= \partial_t\phi + \bfv\cdot\nabla\phi.
\label{eq:Dtphi-ambient}
\end{align}
Thus the material derivative of the pullback on $\omega_t$ agrees with the ambient material derivative of $\phi$ evaluated on the surface.
\end{remark}

\begin{remark}[Jacobi's formula for surface area]
For an arbitrary surface patch $A_t \subset S_t$ and its projection $\omega_t \subset S_w$, we consider the rate of change of the surface area $|A_t|$. Using the Reynolds transport theorem, we have
\begin{equation}
    \frac{\dd}{\dd t}|A_t| = \frac{\dd}{\dd t}\int_{\omega_t}\sqrt{g}~\dd x\dd y = \int_{\omega_t}\left(D_t(\sqrt{g})+\sqrt{g}\nabla_{x,y}\cdot \bfw\right)\dd x \dd y.
\label{eq:rate_of_change_At}
\end{equation}
Using the geometric identity \eqref{eq:geometric_identity_Jacobi_0}, we have
\begin{equation}
    \frac{\dd}{\dd t}|A_t| = \int_{\Omega_t} \sqrt{g}\divS \bfv \dd x \dd y.
\label{eq:geometric_identity_Jacobi}
\end{equation}
Combining \eqref{eq:rate_of_change_At} and \eqref{eq:geometric_identity_Jacobi} and using the fact that $A_t$ is an arbitrary surface patch, we obtain Jacobi's formula for surface area
\begin{equation}
    D_t(\sqrt{g})+\sqrt{g}\nabla_{x,y}\cdot \bfw = \sqrt{g}\divS \bfv.
\label{eq:Jacobi_surface_elements}
\end{equation}
\end{remark}

We now present the graph version of the surface transport theorem.

\begin{lemma}[Surface transport theorem for graph surfaces]
\label{lemma:graph-surface}
Let
$
S_t = \{(x,y,h(x,y,t)):(x,y)\in S_w\}
$
be a smooth evolving graph surface moving with ambient velocity field $\bfv=(v_x,v_y,v_z)$. The planar transport velocity of the projected domain is given by the horizontal component of the ambient velocity evaluated on the graph
$
\bfw=(u,v)=\bigl(v_x,v_y\bigr)\big|_{z=h(x,y,t)}.
$
Let $A_t\subset S_t$ be a material surface patch, let $\omega_t\subset S_w$ be its projection, and let $\phi$ be a scalar function defined in a neighborhood of $S_t$. Define the pullback $\varphi$ by $\varphi(x,y,t) =\phi(x,y,h(x,y,t), t)$.

The Reynolds transport theorem
in graph coordinates states that
\begin{equation}\label{eq:transport-graph}
\frac{\dd}{\dd t}\int_{\omega_t}\varphi \sqrt{g}\,\dd x\dd y
=
\int_{\omega_t}\bigl(D_t(\varphi\,\sqrt{g})+\varphi \sqrt{g}\,\nabla_{x,y}\cdot\bfw\bigr)\,\dd x\dd y.
\end{equation}
The formulation \eqref{eq:transport-graph} is equivalent to
\begin{equation}\label{eq:transport-surface}
\frac{\dd}{\dd t}\int_{A_t}\phi\,\dd\mathcal H^2
=
\int_{A_t}\bigl(\partial_t\phi+\bfv\cdot\nabla\phi+\phi\,\divS\bfv\bigr)\,\dd\mathcal H^2.
\end{equation}
If, in addition, $\bfn\cdot\nabla\phi=0$ on $S_t$, then
\begin{equation}\label{eq:transport-insoluble}
\frac{\dd}{\dd t}\int_{A_t}\phi\,\dd\mathcal H^2
=
\int_{A_t}\bigl(\partial_t\phi+\bfv\cdot\grads\phi+\phi\,\divS\bfv\bigr)\,\dd\mathcal H^2.
\end{equation}
\end{lemma}

\begin{proof}
Using \eqref{eq:area-element} and \eqref{eq:pullback_appdix}, we obtain
\begin{equation}\label{eq:pull-integral}
\int_{A_t}\phi\,\dd\mathcal H^2
=
\int_{\omega_t}\varphi \sqrt{g}\,\dd x\dd y.
\end{equation}
Since $\omega_t$ is a material planar domain transported by $\bfw=(u,v)$, the planar Reynolds transport theorem gives
\begin{equation}\label{eq:reynolds-planar}
\frac{\dd}{\dd t}\int_{\omega_t} \mathcal{G}\,\dd x\dd y
=
\int_{\omega_t}\bigl(\partial_t \mathcal{G}+\nabla_{x,y}\cdot(\mathcal{G}\bfw)\bigr)\,\dd x\dd y
=
\int_{\omega_t}\bigl(D_t \mathcal{G}+\mathcal{G}\,\nabla_{x,y}\cdot\bfw\bigr)\,\dd x\dd y
\end{equation}
for any smooth scalar field $\mathcal{G}$ defined on $\omega_t$.

Applying this identity with $\mathcal{G}=\varphi \sqrt{g}$ yields
\begin{align}
\frac{\dd}{\dd t}\int_{\omega_t}\varphi \sqrt{g}\,\dd x\dd y
&=
\int_{\omega_t}\Bigl(D_t(\varphi \sqrt{g})+\varphi \sqrt{g}\,\nabla_{x,y}\cdot\bfw\Bigr)\,\dd x\dd y \notag\\
&=
\int_{\omega_t}\Bigl(\sqrt{g} D_t\varphi+\varphi(D_t (\sqrt{g})+ \sqrt{g}\nabla_{x,y}\cdot\bfw)\Bigr)\,\dd x\dd y.
\label{eq:expand-product}
\end{align}
Substituting the identities \eqref{eq:Dtphi-ambient} and \eqref{eq:Jacobi_surface_elements} into \eqref{eq:expand-product}, we obtain
\[
\frac{\dd}{\dd t}\int_{\omega_t}\varphi \sqrt{g}\,\dd x\dd y
=
\int_{\omega_t}\Bigl((\partial_t\phi+\bfv\cdot\nabla\phi)\sqrt{g}+\varphi \sqrt{g}\,\divS\bfv\Bigr)\,\dd x\dd y.
\]
Since $\varphi=\phi$ on the surface and $\sqrt{g}\,\dd x\dd y=\dd\mathcal H^2$, we obtain
\[
\frac{\dd}{\dd t}\int_{A_t}\phi\,\dd\mathcal H^2
=
\int_{A_t}\bigl(\partial_t\phi+\bfv\cdot\nabla\phi+\phi\,\divS\bfv\bigr)\,\dd\mathcal H^2,
\]
which proves \eqref{eq:transport-surface}. Formula \eqref{eq:transport-graph} is simply the same identity written in graph coordinates.

If $\bfn\cdot\nabla\phi=0$ on $S_t$, then $\nabla\phi$ is tangent to the surface and hence
\[
\nabla\phi=\grads\phi
\qquad\text{on }S_t.
\]
Therefore
\[
\bfv\cdot\nabla\phi = \bfv\cdot\grads\phi,
\]
and substituting this into \eqref{eq:transport-surface} yields \eqref{eq:transport-insoluble}.
\end{proof}

Next, we consider the special case where the graph surface evolves with purely normal velocity. 
Let $\omega_t \subset \mathbb{R}^2$ denote the time-dependent domain obtained as the projection of a material surface patch $S_t$. Suppose the surface is given as the graph
$
\mathbf r(x,y,t) = (x,y,h(x,y,t))$, where $(x,y)\in \omega_t$,
and evolves with purely normal velocity. We then have the following result.
\begin{lemma}[Surface transport of a graph surface with purely normal velocity]
\label{lemma:surfaceTransport_path_graph}
Let $\phi=\phi(x,y,h(x,y,t), t)$ be a scalar field defined on $S_t$, and let $\varphi(x,y,t) =\phi(x,y,h(x,y,t), t)$ denote its pullback onto $\omega_t$. 
Assume that $S_t$ is a material patch (i.e., it moves with the flow). Then
\begin{align}
\frac{d}{dt}\int_{\omega_t}  \varphi\,\sqrt{g} \,dx\,dy
&=
\int_{\omega_t}
\partial_t \left( \varphi\,\sqrt{g}\right) - \nabla \cdot \left(\varphi  \frac{ h_t   \nabla h}{\sqrt{g}}
\right) dx\,dy \notag\\
=
\int_{\omega_t}
&\left( 
\partial_t \varphi - \frac{ h_t }{\sqrt{g}} \nabla \cdot \left(  \frac{ \varphi\nabla h}{\sqrt{g}} \right)
\right)\,\sqrt{g}\,dx\,dy .
\label{eq:pure_normal_surface_transport}
\end{align}
\end{lemma}

\begin{proof}
In the purely normal case,
\[
\mathbf v = V_n \mathbf n
= \frac{\partial_t h}{g}(-\partial_x h,-\partial_y h,1),
\]
and hence the projected planar velocity is given by
\begin{equation}
\mathbf w = \left(-\frac{\partial_t h \partial_xh}{g},\,
-\frac{\partial_t h \partial_yh}{g}\right).
\label{eq:w_graph}
\end{equation}
Let
$
\mathcal{G} = \varphi\,\sqrt{g}.
$
Applying the Reynolds transport theorem in the plane yields
\begin{equation}
\frac{d}{dt}\int_{S_t} \phi\,dS
=
\frac{d}{dt}\int_{\omega_t} \mathcal{G}\,dx\,dy
=
\int_{\omega_t}
\left(
\partial_t \mathcal{G} + \nabla \cdot (\mathcal{G}\,\mathbf w)
\right) dx\,dy.
\label{eq:identity_normal}
\end{equation}
Substituting the expression for $\mathbf{w}$ in \eqref{eq:w_graph} into \eqref{eq:identity_normal} gives the desired result \eqref{eq:pure_normal_surface_transport}.

\end{proof}

\bibliographystyle{unsrt}
\bibliography{EHD}  

\end{document}